\documentclass[12pt]{extarticle}
\usepackage{amsmath, amsthm, amssymb, color}
\usepackage[colorlinks=true,linkcolor=blue,urlcolor=blue]{hyperref}
\usepackage{graphicx}
\usepackage{caption}
\usepackage{mathtools}
\usepackage{enumerate}
\usepackage{verbatim}
\usepackage{tikz,tikz-cd,tikz-3dplot}
\usepackage{amssymb}
\usetikzlibrary{matrix}
\usetikzlibrary{arrows}
\usepackage{algorithm}
\usepackage[noend]{algpseudocode}
\usepackage{caption}
\usepackage[normalem]{ulem}
\usepackage{subcaption}
\usepackage{cite}

\oddsidemargin \evensidemargin
\usepackage{makecell}
\usepackage{array}
\usepackage[scr=rsfs]{mathalpha}
\usepackage{multirow,array}

\DeclareFontEncoding{LS1}{}{}
\DeclareFontSubstitution{LS1}{stix}{m}{n}
\DeclareSymbolFont{symbols2}{LS1}{stixfrak} {m} {n}
\DeclareMathSymbol{\operp}{\mathbin}{symbols2}{"A8}

\newtheorem{theorem}{Theorem}
\newtheorem{proposition}[theorem]{Proposition}
\newtheorem{lemma}[theorem]{Lemma}
\newtheorem{corollary}[theorem]{Corollary}

\theoremstyle{definition}

\newtheorem{remark}[theorem]{Remark}

\newtheorem{example}[theorem]{Example}

\numberwithin{theorem}{section}

\title{Algebraic Geometry for Spin-Adapted Coupled Cluster Theory}
\author{Fabian M.\ Faulstich and Svala Sverrisdóttir}
\date{\today}

\newcommand{\PP}{\mathbb{P}}

\newcommand{\CC}{\mathbb{C}}
\newcommand{\ZZ}{\mathbb{Z}}

\usepackage{xcolor}

\begin{document}

\maketitle

\begin{abstract}
We develop and numerically analyze an algebraic--geometric framework for spin-adapted coupled-cluster (CC) theory. Since the electronic Hamiltonian is ${\rm SU}(2)$--invariant, physically relevant quantum states lie in the spin singlet sector. We give an explicit description of the ${\rm SU}(2)$--invariant (spin singlet) many-body space by identifying it with an Artinian commutative ring, called the excitation ring, whose dimension is governed by a Narayana number. We define spin--adapted truncation varieties via embeddings of graded subspaces of this ring, and we identify the CCS truncation variety with the Veronese square of the Grassmannian. Compared to the spin-generalized formulation, this approach yields a substantial reduction in dimension and degree, with direct computational consequences. In particular, the CC degree of the truncation variety -- governing the number of homotopy paths required to compute all CC solutions -- is reduced by orders of magnitude. We present scaling studies demonstrating asymptotic improvements and we exploit this reduction to compute the full solution landscape of spin-adapted CC equations for water and lithium hydride.
\end{abstract}

\section{Introduction}

Interacting many-electron problems pose some of the greatest computational challenges in science, with essential applications across chemistry, materials science, and condensed-matter physics~\cite{wang2024ab,aizawa2022delayed,kirkpatrick2021pushing,cui2022systematic,chang2024downfolding,herzog2024coupled,hagai2024extended,jing2023precise,shi2025accurate}. Scalable accurate solutions will enable predictive simulations of chemical reactivity and kinetics, as well as reliable access to ground- and excited-state properties of quantum systems. Coupled cluster (CC) theory is widely regarded as one of the most accurate general-purpose methods in electronic structure theory~\cite{bartlett2007coupled,crawford2007introduction,kummel1991origins,vcivzek1991origins,bartlett2005theory,paldus2005beginnings,arponen1991independent,bishop1991overview}. It is routinely used as a benchmark in method development~\cite{grafova2010comparative,takatani2010basis,jurevcka2006benchmark,rezac2011s66,rezac2012benchmark} and as a target accuracy level for data-driven and machine-learning approaches~\cite{ramakrishnan2015big,smith2019approaching,smith2020ani,ruth2022machine,daru2022coupled}. 

\medskip
The governing equations, called the {\it CC equations}, correspond to a system of polynomial equations~\cite{vcivzek1966correlation,purvis1982full,noga1987full}. Under certain assumptions, isolated roots of this system, called the {\it CC solutions}, can recover the exact eigenstates of the underlying electronic Hamiltonian within the chosen model space. Understanding the structure of this solution set, e.g., its size, multiplicities, and dependence on physical parameters, is therefore of computational as well as physical interest. Yet, because of the high dimensionality and strong nonlinearity of the CC equations, many fundamental features of their solution landscapes remain undiscovered.

\medskip
In this work, we extend recent algebraic investigations establishing a foundational understanding of the coupled-cluster solution set~\cite{faulstich2024coupled, faulstich2024algebraic, sverrisdottir2024exploring}. In
particular, we here study \emph{spin adaptation} and its effect on the structure and complexity of the CC equations~\cite{paldus1977correlation}. Spin adaptation is
central in quantum many-body applications: the conservation of total spin implies that the physically relevant states lie in prescribed
spin sectors, e.g., singlets for closed-shell molecules~\cite{pauncz2012spin}. Restricting the parameterization to spin-adapted excitations both enforces physical consistency and substantially simplifies the associated polynomial system, leading to significant computational speedups confirmed by our numerical results.

\medskip
We develop an algebraic-geometric formulation of spin-adapted coupled cluster theory and use it to reduce the complexity of {\it fully} solving the CC equations in the physically relevant spin singlet sector. In Section \ref{sec:repsFD} we study the
${\rm SU}(2)$--invariant space of the of $d$--electron space $\mathcal{H}_d \cong \wedge^d \mathcal{H}$. We derive a closed formula for the dimensions of the total-spin sectors, including the ${\rm SU}(2)$--invariant subspace. We identify $\mathcal{H}_d$ with a particle--number conserving subalgebra of the Fermi--Dirac algebra. There, spin adaptation becomes the restriction to ${\rm SU}(2)$--invariant cluster amplitudes. We prove that the space of invariant amplitudes is an Artinian commutative ring, called the excitation ring, generated by the \emph{excitation operators}. In Section \ref{sec:cctheory} we introduce the \emph{spin singlet truncation varieties}. They are the Zariski closures of the projective images of graded subspaces of the excitation ring under the exponential map. We show that the singlet CCS truncation variety, consisting of ${\rm SU}(2)$--invariant Slater determinants, is the Veronese square of the Grassmannian. The truncation varieties provide a geometric framework for formulating the spin-adapted CC equations.
We offer an upper bound of the CC degree (general solution count for the CC equations) in terms of invariants of the truncation varieties; and we find the explicit CC degree for the $2$nd Veronese variety. 

\medskip
In Section \ref{sec:numsim} we illustrate the computational improvements of these structural results. The CC degree determines the number of continuation paths in homotopy-based
solvers, therefore restricting to ${\rm SU}(2)$--invariants collapses the generic root count while targeting {\it spin-pure solutions}. We compare the CC degrees of the spin-adapted CC equations with those of the spin-generalized formulation, revealing dramatic reductions by orders of magnitude. We compute the full solution spectrum of the singlet CCD equations for lithium hydrate (LiH) and water (H$_2$O) in a minimal-basis.
Numerically, the speedup is already substantial for small systems, enabling systematic \emph{algebro-computational} investigations of coupled-cluster solution landscapes that were computationally out of reach in the spin-generalized formulations.

\subsection{Previous Works and Perspective}

The first systematic study of the root structure of coupled cluster (CC) equations dates back to 1978, when \v{Z}ivkovi\'c and Monkhorst analyzed singularities and multiple solutions in single-reference CC theory~\cite{vzivkovic1978analytic}. In the early 1990s, Paldus and coauthors developed further mathematical and numerical analyses of the solution manifolds of single-reference and state-universal multireference CC equations, including their singularities and analytic properties~\cite{paldus1993application,piecuch1990coupled}. Homotopy methods were revived in the CC context by Kowalski and Jankowski, who solved a CCD system in~\cite{kowalski1998towards}; Kowalski and Piecuch subsequently extended these ideas to CCSD, CCSDT, and CCSDTQ for a four-electron minimal-basis model~\cite{piecuch2000search}. These works led to the $\beta$-nested equations and the \emph{Fundamental Theorem of the $\beta$-NE Formalism}~\cite{kowalski2000complete2}, which explains how solution branches connect across truncation levels (e.g., CCSD $\to$ CCSDT $\to$ CCSDTQ). The resulting \emph{Kowalski--Piecuch homotopy} has recently been analyzed in depth using topological degree theory~\cite{csirik2023disc,csirik2023coupled}. Related homotopy-based studies produced complete solution sets of the generalized Bloch equation~\cite{kowalski2000complete} and clarified symmetry-breaking mechanisms within CC theory~\cite{kowalski1998full,jankowski1999physical1,jankowski1999physical2,jankowski1999physical3,jankowski1999physical4}.

\medskip
A more explicitly algebraic and mathematically driven line of work has recently gained momentum. Initiated by some of the present authors, an algebraic perspective on CC theory was introduced in \cite{faulstich2024coupled}, and subsequent in-depth analysis led to the algebraic framework of CC truncation varieties~\cite{faulstich2024algebraic}. This viewpoint has since motivated a sequence of contributions ranging from pure mathematics~\cite{borovik2025coupled,sverrisdottir2025algebraic} to domain applications~\cite{sverrisdottir2024exploring}.
More broadly, these developments reflect a growing interest in a rigorous mathematical understanding of CC theory within the pure and applied mathematics community. In this context, Schneider gave the first local analysis of CC theory in 2009 using functional-analytic techniques~\cite{schneider2009analysis}; this framework was generalized in~\cite{rohwedder2013continuous,rohwedder2013error} and extended to additional CC variants in~\cite{laestadius2018analysis,laestadius2019coupled,faulstich2019analysis}. A complementary and more flexible approach based on topological degree theory was later proposed by Csirik and Laestadius~\cite{csirik2023disc,csirik2023coupled}. The most recent numerical-analysis results for single-reference CC were obtained by Hassan, Maday, and Wang~\cite{hassan2023analysis,hassan2023analysis2}. For a comprehensive review of the recent mathematical advances, we refer the interested reader to~\cite{faulstich2024recent}. Beyond the direct applications to coupled cluster theory, algebro computational approaches find more and more use in electronic structure theory~\cite{gontier2019numerical,cances2024mathematical,pokhilko2025homotopy,harwood2022improving,borovik2025numerical}

\subsection{Notational Preamble}

Throughout this manuscript, $m$ denotes the number of spatial orbitals and $n=2m$ the number of spin orbitals.
We generally consider a closed-shell system with $d=2k$ electrons ($k$ electron pairs).
We use $[m]:=\{1,...,m\}$ and denote $[\![m]\!]:=[m]\times\{\uparrow,\downarrow\}$. 
Operators in the Fermi--Dirac algebra  ${\rm FD}_{2m} \cong{\rm End}(\mathcal F)$ carry hats.

\section{The Electronic Structure Problem}

The {\em electronic Schr\"odinger equation} describing $d$ electrons in the electromagnetic field generated by $d_{\rm nuc}$ fixed nuclei is the eigenvalue problem
\begin{equation}
\label{eq:schroedinger}
\mathscr{H}\,\psi(\mathbf{x}_1,\mathbf{x}_2,\ldots,\mathbf{x}_d)
=
\lambda\,\psi(\mathbf{x}_1,\mathbf{x}_2,\ldots,\mathbf{x}_d),
\end{equation}
where the unknown {\em wave function} $\psi(\mathbf{x}_1,\mathbf{x}_2,\ldots,\mathbf{x}_d)$ fully characterizes the physical system~\cite{schrodinger1926undulatory}.
Here $\mathbf{x}_i=(\mathbf{r}_i,\sigma_i)\in\mathbb{R}^3\times\{\uparrow, \downarrow \}$ denotes the spatial coordinate $\mathbf{r}_i$ and spin coordinate $\sigma_i$ of the $i$th electron. The Pauli principle requires $\psi$ to be skewsymmetric under exchange of any two electron coordinates (vide infra), and we view $\psi$ as a sufficiently smooth element of the skewsymmetric subspace of
$L^2\!\big((\mathbb{R}^3\times\{\uparrow, \downarrow \})^d\big)$. 
In atomic units, the electronic Hamiltonian $\mathscr{H}$ in Eq.~\eqref{eq:schroedinger} is the differential operator
\begin{equation}
\label{eq:ElecSE}
\mathscr{H}
:=
-\frac{1}{2}\sum_{i=1}^d\Delta_{\mathbf{r}_i}
-\sum_{i=1}^d\sum_{j=1}^{d_{\rm nuc}}\frac{Z_j}{\vert \mathbf{r}_i-\mathbf{R}_j\vert}
+\sum_{1\le i<j\le d}\frac{1}{\vert \mathbf{r}_i-\mathbf{r}_j\vert}.
\end{equation}
We denote  the Laplacian in
$\mathbf{r}_i$ by $\Delta_{\mathbf{r}_i}$, and all remaining terms in Eq.~\eqref{eq:ElecSE} act on $\psi$ by multiplication. The nuclei are indexed by $j\in[d_{\rm nuc}]$. The constant $Z_j\in\mathbb{N}$ is the charge (atomic number)
of the $j$th nucleus, and its position $\mathbf{R}_j\in\mathbb{R}^3$ is fixed under the Born--Oppenheimer approximation~\cite{schrodinger1927szur}. Throughout, we consider charge-neutral molecules, i.e.,
$
d=\sum_{j=1}^{d_{\rm nuc}} Z_j.
$
Note that any nucleus--nucleus repulsion term is a constant for fixed $\{\mathbf{R}_j\}$ and may be added to $\mathscr{H}$ without changing the structure of the electronic eigenvalue problem.

\medskip
Electrons possess an intrinsic angular momentum called \emph{spin}~\cite{sakurai2020modern}. In the non-relativistic
Born--Oppenheimer model considered here, spin does not appear explicitly in the differential operator in Eq.~\eqref{eq:ElecSE}. However, it enters the \emph{state space} and the admissibility condition (Pauli principle) in the following way: The physical $d$-electron wave function is a map
\begin{equation}
\label{eq:spinwavefunction}
\psi : (\mathbb{R}^3)^d \times\{\uparrow, \downarrow \}^d \to \mathbb{C}
~;~
(\mathbf r,\boldsymbol\sigma)\mapsto \psi(\mathbf r,\boldsymbol\sigma),
\end{equation}
i.e., it depends on spatial coordinates $\mathbf r=(\mathbf r_1,\ldots,\mathbf r_d)$ and spin labels
$\boldsymbol\sigma=(\sigma_1,\ldots,\sigma_d)$. Pauli's exclusion principle~\cite{pauli1925zusammenhang} states that admissible states are \emph{skewsymmetric under simultaneous exchange} of position and spin, i.e., for every permutation $P$ of the electrons,
\begin{equation}
\label{eq:pauli}
\psi(P\mathbf r, P\boldsymbol\sigma) = \mathrm{sign}(P)\,\psi(\mathbf r,\boldsymbol\sigma).
\end{equation}
In particular, two electrons with equal spin cannot occupy the same position. 

\medskip
To discretize the Hamiltonian in Eq.~\eqref{eq:schroedinger}, we require a finite-dimensional basis. Since the Hamiltonian does not explicitly depend on the electronic spin, we can impose a product basis. To that end, we first introduce the one-electron \emph{spin space}
$
\mathcal H_{\mathrm{spin}} \cong \mathbb{C}^2
$
with orthonormal basis $\{e_\uparrow,e_\downarrow\}$ (spin up/down). Second, we introduce an orthonormal set of spatial basis functions (orbitals)
$\{\varphi_p\}_{p=1}^{m}\subset H^1(\mathbb{R}^3)$.
Each spatial orbital gives rise to two spin-orbitals
\begin{equation}
\phi_{p\uparrow}(\mathbf r,\sigma):=\varphi_p(\mathbf r)\,\delta_{\sigma,\uparrow},
\qquad
\phi_{p\downarrow}(\mathbf r,\sigma):=\varphi_p(\mathbf r)\,\delta_{\sigma,\downarrow}.
\end{equation}
The spin-orbitals span the one-particle space $\mathcal H$, with
$\mathcal H\cong \mathbb C^{m}\otimes\mathbb C^{2}$ and~$\dim(\mathcal H)~=~n~=~2m$. The Galerkin space of skewsymmetric $d$-electron functions is given by $\mathcal H_d \cong \wedge^d\mathcal H$, spanned by the canonical skewsymmetric product basis, i.e., the $d$-electron Slater determinants. 

\subsection{Second Quantization and Fermi Dirac Algebra}
\label{sec:FermiDiracAlgebra}

Discretizing the Hamiltonian in the skewsymmetric $d$-electron space $\wedge^d\mathcal H$ is known as \emph{first quantization}. While conceptually direct, the resulting operators act on a high-dimensional tensor product and quickly become cumbersome to manipulate. An equivalent and more flexible viewpoint is obtained by working in the fermionic Fock space and representing observables as polynomials in creation and annihilation operators, i.e., \emph{second quantization}~\cite{fock1932konfigurationsraum}.

\medskip
Subsequently, we use the explicit basis $e_\downarrow = -e_2, \, e_\uparrow = e_1$ for $\mathcal{H}_{{\rm spin}}$, where $e_1$ and $e_2$ form the standard basis of $\mathbb{C}^2$ and we denote the standard basis vectors of $\mathbb{C}^m$ by $e_p$ with $p\in [m]$. For notational convenience, we define 
$
[\![m]\!] := [m]\times\{\downarrow,\uparrow\}
$, with $n = 2m$ elements. Hence, a basis of $\mathcal{H}$ is given by $e_{p\alpha}:=e_p\otimes e_\alpha$, $(p,\alpha)\in[\![m]\!]$.
The associated fermionic Fock space is the full exterior algebra
\begin{equation}
\mathcal{F}\;\cong\;\bigwedge \mathcal{H}
\;=\;
\bigoplus_{\ell=0}^{n}\bigwedge\nolimits^{\ell}\mathcal{H},
\end{equation}
with canonical exterior basis vectors $e_I:=\wedge_{(p,\alpha)\in I} e_{p\alpha}$ indexed by sets $I\subseteq [\![m]\!]$. For each basis element $e_{p\alpha}\in \mathcal{H}$ we define the creation and annihilation operators
as the following exterior and interior products on $\mathcal{F}$:
\begin{equation}
\label{eq:CreandAni}
\begin{aligned}
\hat a_{p\alpha}^\dagger : \mathcal{F} \to \mathcal{F}~;~
\psi \mapsto e_{p\alpha}\wedge \psi
\quad  {\rm and} \quad 
\hat a_{p\alpha} : \mathcal{F} \to \mathcal{F}~;~
\psi \mapsto e_{p\alpha}\,\lrcorner\,\psi.
\end{aligned}
\end{equation}
Here, $\lrcorner$ is the dual operation to the wedge product $\wedge$, see \cite[Section~3.6]{gallier2020differential} and~\cite{sverrisdottir2025algebraic} for further explanation.  
These operators satisfy the canonical anticommutation relations (CAR), i.e., 
\begin{equation}
\label{eq:fermirels}
\begin{aligned}
[\hat{a}_{p\alpha}^\dagger, \hat{a}_{q\beta}^\dagger]_+ = 0,\quad
[\hat{a}_{p\alpha}, \hat{a}_{q\beta}]_+ =0,\quad
[\hat{a}_{p\alpha}^\dagger, \hat{a}_{q\beta}]_+ = \delta_{p,q}\delta_{\alpha,\beta},
\end{aligned}
\end{equation}
where $[\cdot, \cdot]_+$ denotes the anti-commutator. Note that the CAR implement Pauli's exclusion principle at the operator level, since $(\hat{a}_{p\alpha}^\dagger)^2 = 0$. 

\medskip
The \emph{Fermi--Dirac algebra} ${\rm FD}_{2m}$ is the noncommutative algebra generated by the
$4m$ creation and annihilation operators
subject to the CAR in Eq.~\eqref{eq:fermirels}. Equivalently,
\[
{\rm FD}_{2m}
:=
\mathbb{C}\!\left\langle
\hat{a}_{1\downarrow},\hat{a}_{1\uparrow},\ldots,\hat{a}_{m\downarrow},\hat{a}_{m\uparrow}, \hat{a}_{1\downarrow}^\dagger,\hat{a}_{1\uparrow}^\dagger,\ldots,\hat{a}_{m\downarrow}^\dagger,\hat{a}_{m\uparrow}^\dagger
\right\rangle \Big/ \langle G\rangle,
\]
where $G$ denotes the set of relations from Eq.~\eqref{eq:fermirels}.
Using the creation and annihilation operators, we can discretize the Hamiltonian in Eq.~\eqref{eq:ElecSE} over $\mathcal{F}$~\cite{sakurai2020modern,helgaker2013molecular}, i.e., 
\begin{equation}
\label{eq:Hsecondquant}
\hat H = \sum_{pq,\sigma} h_{pq} \, \hat{a}_{p\sigma}^\dagger \hat{a}_{q\sigma}
+ \tfrac{1}{2} \sum_{pqrs}\sum_{\sigma\tau} v_{pqrs} \,
\hat{a}_{p\sigma}^\dagger \hat{a}_{r\tau}^\dagger \hat{a}_{s\tau} \hat{a}_{q\sigma},
\end{equation}
where $h_{pq}$ and $v_{pqrs}$ are one- and two-electron integrals over spatial orbitals,
\begin{align}
h_{pq} &:= \int \varphi_p^*(\mathbf{r}) \left( -\frac{1}{2}\Delta - \sum_{j} \frac{Z_j}{\Vert \mathbf{r} - \mathbf{R}_j \Vert}\right) \varphi_q(\mathbf{r}) \, d\mathbf{r}, \\
v_{pqrs} &:= \iint \varphi_p^*(\mathbf{r}_1)\varphi_q(\mathbf{r}_1)
\frac{1}{\Vert \mathbf{r}_1-\mathbf{r}_2\Vert }
\varphi_r(\mathbf{r}_2)^*\varphi_s(\mathbf{r}_2) \, d\mathbf{r}_1\, d\mathbf{r}_2.
\end{align}
Since its monomials contain equal numbers of creation and annihilation operators, $\hat H$ conserves particle number and may be restricted to an endomorphism of the $d$th exterior power
\begin{equation}
\mathcal{H}_d \;\cong\; \bigwedge\nolimits^{d}\mathcal{H}\;\subset\;\mathcal{F}.
\end{equation}

\begin{remark}
Wick's theorem states that $G$ forms a Gröbner basis for the Fermi--Dirac algebra, with respect to degree lexicographic order~\cite[Theorem~2.5]{sverrisdottir2025algebraic}. Therefore, the Fermi--Dirac algebra is isomorphic to a $2^{4m}$-dimensional vector space generated by the \textit{standard monomials}, also known as the normal ordered strings of ladder operators~\cite{dirac1927quantum,JordanWigner1928}. Every element in FD$_{2m}$ can be written as a linear combination of them, called the \textit{standard representation}.
\end{remark}

\medskip
Elements of the Fermi--Dirac algebra are endomorphisms of $\mathcal{F}$ and, by dimension counting, we see that, as a vector space, the Fermi--Dirac algebra is isomorphic to $\operatorname{End}(\mathcal{F})$. 
Moreover, the Fermi--Dirac algebra emits a natural $\ZZ^{2}$--grading, by setting $\deg(a_{p\alpha}^\dagger) = e_1$ and $\deg(a_{p\alpha}) = e_2$. This bidegree coincides with the \emph{annihilation and creation levels} introduced in \cite[Section~3]{sverrisdottir2025algebraic}.
Subsequently, we assume particle--number conservation and restrict attention to the subalgebra generated by elements of bidegree $(\ell,\ell)$ for $1 \le \ell \le m$, which we refer to as having \textit{excitation level} $\ell$.
As a vector space, this subalgebra is isomorphic to
\begin{equation}
\operatorname{End}(\mathcal H_0)\oplus \operatorname{End}(\mathcal H_1)\oplus \cdots \oplus \operatorname{End}(\mathcal H_n),
\qquad \mathcal H_\ell:=\wedge^\ell \mathcal H.
\end{equation}

\section{Representations of the Fermi--Dirac Algebra}\label{sec:repsFD}

The electronic spin can be viewed as a two-dimensional internal degree of freedom of each electron~\cite{pauncz2012spin,sakurai2020modern}. In the non-relativistic Born--Oppenheimer setting (i.e., in the absence of spin--orbit and magnetic-field terms), the electronic-structure Hamiltonian $\hat H$ in Eq.~\eqref{eq:Hsecondquant} is $\mathrm{SU}(2)$-invariant (i.e., spin-rotationally invariant) and hence commutes with the total-spin operators $\hat S^2$ and $\hat S_z$ (vide infra), so that the Hilbert space decomposes into spin sectors (e.g., singlets, triplets, etc.)~\cite{pauncz2012spin}. For spin-$\tfrac12$ the relevant symmetry group is not $\mathrm{SO}(3)$ itself but its double cover $\mathrm{SU}(2)$, whose defining representation acts on the spin space $\mathcal H_{\mathrm{spin}}\cong\mathbb C^2$~\cite{pauncz2012spin,sakurai2020modern,takhtadzhian2008quantum}. To exploit this symmetry in many-electron computations, we study the induced $\mathrm{SU}(2)$-action on the skewsymmetric $d$-electron space and, equivalently, on the operator algebra in second quantization. This provides a precise algebraic notion of spin invariance and allows us to isolate, in particular, the singlet sector used in spin-adapted formulations~\cite{paldus1972correlation}. The argument uses two standard facts. First, differentiation induces a one-to-one correspondence between finite-dimensional
representations of $\mathrm{SU}(2)$ and finite-dimensional representations of $\mathfrak{su}(2)$. Indeed, because  $\mathrm{SU}(2)$ is
simply connected, every finite-dimensional $\mathfrak{su}(2)$--representation integrates to $\mathrm{SU}(2)$
(cf.~\cite[Exercise~8.10]{fulton2013representation}). Second, complex representations of $\mathfrak{su}(2)$ are in
one-to-one correspondence with representations of $\mathfrak{sl}_2(\mathbb C)$ via complexification, since
$\mathfrak{su}(2)\otimes_{\mathbb R}\mathbb C \cong \mathfrak{sl}_2(\mathbb C)$.

\subsection{Spin of a Single Electron}
Recall that the Lie algebra $\mathfrak{su}(2)$ consists of traceless skew-Hermitian $2\times 2$ matrices,
\begin{equation}
\mathfrak{su}(2)=\{A\in\mathbb{C}^{2\times 2}:\ A^\ast=-A,\ {\rm tr}(A)=0\}.
\end{equation}
The three $2 \times 2$ matrices $( \tfrac{i}{2}\sigma_x, \tfrac{i}{2}\sigma_y, \tfrac{i}{2}\sigma_z)$, where $\sigma_x,\sigma_y,\sigma_z$ are the Pauli matrices~\cite{pauli1927quantenmechanik}, form a basis for $\mathfrak{su}(2)$. Similarly, it is common to work with the associated Hermitian operators, i.e., $S_x=\tfrac12\sigma_x$, $S_y=\tfrac12\sigma_y$, and $S_z=\tfrac12\sigma_z$, noting that $-iS_x,-iS_y,-iS_z\in\mathfrak{su}(2)$ and that spin rotations are given by $\exp(-i\,\theta_x S_x-i\,\theta_y S_y-i\,\theta_z S_z)\in{\rm SU(2)}$. For our purposes, it is convenient to work with the associated spin ladder operators, i.e., 
\begin{equation}
S_+ 
= S_x + i S_y
= \frac{1}{2}(\sigma_x + i\sigma_y),\qquad
S_- 
= S_x - i S_y
= \frac{1}{2}(\sigma_x - i\sigma_y),\qquad
S_z = \frac{1}{2}\sigma_z,
\end{equation}
where $S_z$ is diagonal in the $\{\uparrow,\downarrow\}$ basis and hence directly encodes the spin projection, while $S_+$ and $S_-$ act as raising and lowering operators, i.e., flipping ``$\downarrow~\mapsto~\uparrow$'' and ``$\uparrow~\mapsto~\downarrow$'', respectively. Note that $S_\pm$ are not elements of $\mathfrak{su}(2)$ but lie in its complexification $\mathfrak{sl}_2(\mathbb C)$. Moreover, a central operator for the discussion of spin symmetry is the {\it total--spin operator} 
\begin{equation}
\label{eq:Casimir}
S^2
:= S_x^2 + S_y^2 + S_z^2
= S_z^2+\tfrac12\big(S_+S_-+S_-S_+\big).
\end{equation}

\begin{remark}\label{re:irrreps}
    The finite-dimensional irreducible representations of ${\rm SU}(2)$ are parametrized by non-negative half-integers: for each $j \in \tfrac{1}{2}\mathbb{Z}_{\ge 0}$ there is an irreducible representation $V_j \cong \operatorname{Sym}^{2j}(\mathbb{C}^2)$, see \cite[Section~11.1]{fulton2013representation}.
    The space $\operatorname{Sym}^{2j}(\mathbb{C}^2)$ can be realized as the space of homogeneous polynomials in two variables $x,y$ of degree $2j$. It has the monomial basis $\{x^{2j-i} y^i \mid 0 \le i \le 2j\}$, and hence $\dim V_j = 2j + 1$. 
    The ${\rm SU}(2)$--action on $V_j$ is the induced action on homogeneous polynomials coming from left multiplication of the variable vector.
\end{remark}

The irreducible representation of ${\rm SU}(2)$ on $V_{\frac{1}{2}} \cong \mathbb{C}^2$ is called the spin-$\tfrac{1}{2}$ representation. Differentiation yields the corresponding representation of the Lie algebra $\mathfrak{su}(2)$, commonly referred to as the fundamental representation. We define an $\mathfrak{su}(2)$--action on the spin space  $\mathcal{H}_{\rm spin} \cong \mathbb{C}^2 \cong V_{\frac{1}{2}}$ by describing the action of the generators on the basis vectors, i.e.\

\begin{equation}
\label{eq:su2action2by2}
    \begin{aligned}
    S_{+}(e_{\downarrow}) &= -e_{\uparrow}& \quad
    S_{-}(e_{\downarrow}) &= 0& \quad 
    S_{z}(e_{\downarrow}) &= -\frac{1}{2}e_{\downarrow}& \\  
    S_{+}(e_{\uparrow}) &= 0& \quad 
    S_{-}(e_{\uparrow}) &= -e_{\downarrow}& \quad 
    S_{z}(e_{\uparrow}) &= \frac{1}{2}e_{\uparrow}.&
\end{aligned}
\end{equation}
See equations (3.10)--(3.12) in \cite[Section~3.2]{takhtadzhian2008quantum}. 
In the explicit basis $e_{\downarrow} = -e_2$,  $e_{\uparrow} = e_1$ this action coincides with left multiplication: for $g \in \mathfrak{su}(2)$ it acts on $\mathcal{H}_{\rm spin}$ by $v \mapsto gv$.

\medskip
We can extend this action to an $\mathfrak{su}(2)$--action on the one-particle space $\mathcal H$ by acting only on the
spin factor, i.e., $g(e_{p\alpha}) = e_p \otimes g e_\alpha$ for $g \in \mathfrak{su}(2)$.
Equivalently, in the basis $\{e_{p\alpha}\}_{(p,\alpha)\in[\![m]\!]}$ the action is represented by
\begin{equation}
g_m := I_m \otimes g \in \mathbb C^{(2m)\times(2m)} \cong \mathbb C^{n\times n}.
\end{equation}
Note that this representation is reducible.

\subsection{Spin of $d$-Electron States and Spin-Singlet $d$-Electron States}
We further induce the $\mathfrak{su}(2)$--action onto the space $\mathcal{H}_d$ of $d$-particle states, canonically via
\begin{equation}
\label{eq:wedgeaction}
g(e_I) 
:=
\sum_{r = 1}^d e_{i_1\alpha_1} \wedge \cdots \wedge g(e_{i_r \alpha_r}) \wedge \cdots \wedge e_{i_d\alpha_d}.
\end{equation}
In the exterior basis $\{e_I:|I|=d\}$ of $\mathcal H_d$, the action is represented by a $\binom{n}{d}\times\binom{n}{d}$ matrix acting by left multiplication on the coordinate vectors:
\begin{equation}
\label{eq:wedgeaction_insertion}
g^{(d)}
:=
\sum_{k=1}^d
\underbrace{I_{\mathcal{H}}\wedge ... \wedge I_{\mathcal{H}}}_{k-1~{\rm times}}
\wedge
g_m
\wedge 
\underbrace{I_{\mathcal H}\wedge ... \wedge  I_{\mathcal{H}}}_{d-k~{\rm times}}
\;\in\;{\rm End}(\mathcal H_d).
\end{equation}

The space of {\it total-spin singlet states} (or simply {\it spin singlets}) $\mathcal H_d^{\rm SU(2)}$ is defined by the quantum states that are invariant under the natural SU(2)-action on the spin factor, or equivalently, the quantum states that are mapped to zero by the derived $\mathfrak{su}(2)$--action (cf.~\eqref{eq:wedgeaction}).  We can describe this ${\rm SU}(2)$-invariant subspace explicitly by

\begin{equation}
\label{eq:SpinSinglets}
\mathcal H_d^{\rm SU(2)}
:=\{\psi\in\mathcal H_d~:~ S_z^{(d)}\psi=S_+^{(d)}\psi=S_-^{(d)}\psi=0\}.
\end{equation}
Alternatively, we define the space of spin singlets as the eigenspace of $S^2$ with eigenvalue $0$:
\begin{equation}
\label{eq:SpinSinglets2}
\mathcal H_d^{\rm SU(2)}
=\{\psi\in\mathcal H_d~:~ \left(S^{(d)}\right)^2 \psi=0\}.
\end{equation}
We see that by definition of $\left(S^{(d)}\right)^2$, Eq.~\eqref{eq:SpinSinglets} implies Eq.~\eqref{eq:SpinSinglets2}. For the inverse, recall that $S_x^{(d)}$, $S_y^{(d)}$, and $S_z^{(d)}$ are self-adjoint and therefore 
\begin{equation}
\begin{aligned}
\langle \psi ,\left(S^{(d)}\right)^2 \psi \rangle
&=
\langle \psi , \left(S_x^{(d)}\right)^2 \psi \rangle
+ \langle \psi , \left(S_y^{(d)}\right)^2 \psi \rangle
+ \langle \psi , \left(S_z^{(d)}\right)^2 \psi \rangle\\
&= \Vert S_x^{(d)} \psi \Vert^2 +  \Vert S_y^{(d)} \psi \Vert^2 + \Vert S_z^{(d)} \psi\Vert^2 \geq 0.
\end{aligned}
\end{equation}
Hence, $\psi \in {\rm ker}(\left(S^{(d)}\right)^2)$ implies $S_x^{(d)}\psi=S_y^{(d)}\psi=S_z^{(d)}\psi=0$ and consequently $S_\pm^{(d)}\psi=0$. More generally, the total-spin operator $\left(S^{(d)}\right)^2$ is the Casimir operator of the induced $\mathfrak{su}(2)$--action~\cite{GreinerMueller1994QMSymmetries}. Consequently, it acts on $\mathcal{H}_d$ by multiplication of a scalar. In particular, we define the spin-$s$ sector as the subspace of $\mathcal{H}_d$ where
\begin{equation}
\left(S^{(d)}\right)^2 \psi = s(s+1)\,\psi.
\end{equation}
For $s=0, \tfrac12, 1, \tfrac32,\dots$, these spaces are referred to as the spin singlet, doublet, triplet, quartet, and higher spin sectors, respectively.
For example, the spin triplet sector is the eigenspace of $\left(S^{(d)}\right)^2$ corresponding to eigenvalue $2$. This yields a canonical decomposition of $\mathcal H_d$ into irreducible spin sectors.
The electronic Hamiltonian in Eq.~\eqref{eq:Hsecondquant} is ${\rm SU}(2)$--invariant and hence it commutes with $\left(S^{(d)}\right)^2$, and therefore preserves the decomposition. Computations are thus typically carried out within a fixed sector -- in this paper, the spin singlet sector.

\medskip
From a representation-theoretic viewpoint, the singlet sector corresponds to the trivial representation $V_0$ of ${\rm SU}(2)$, which is $1$-dimensional and consists of ${\rm SU}(2)$--invariant vectors (equivalently, vectors annihilated by the derived $\mathfrak{su}(2)$--action). The space $\mathcal H_d$ is completely reducible as an ${\rm SU}(2)$--module. Hence it decomposes into irreducibles with multiplicities,
\begin{equation}
\mathcal{H}_d \cong \bigoplus_{j\in \frac12\mathbb Z_{\ge0}}  V_j^{\oplus c_j},
\end{equation}
where $V_j$ denotes the $j$th irreducible ${\rm SU}(2)$--representation, defined in Remark \ref{re:irrreps}. In this decomposition, the spin singlet sector is exactly the $j=0$ isotypic component,
\begin{equation}
\mathcal{H}_d^{{\rm SU}(2)} \cong  V_0^{c_0} \cong \mathbb{C}^{c_0}.
\end{equation}
Equivalently, the spin-$s$ sectors are isomorphic to the $s$ isotypic comonents $V_s^{\oplus c_s}$~\cite{GreinerMueller1994QMSymmetries}.

\begin{remark}
\label{rmk:SpinRestricted}
In quantum chemistry, ``spin-restricted'' refers to a constraint at the one-particle (orbital) level, e.g.,~in spin-restricted Hartree-Fock theory, the $\alpha$ and $\beta$ electrons share the same spatial orbitals. In particular, spin-restricted is primarily a statement about the reference and orbital parametrization, not about the symmetry content of a correlated ansatz. By contrast, \emph{spin-adapted} (i.e., SU$(2)$-adapted) refers to the \emph{many-electron} level: the wavefunction (or equivalently the operator ansatz acting on a reference) is constructed to transform in a fixed irreducible SU$(2)$--representation. Thus, ``restricted'' and ``adapted'' address different levels of theory (i.e., orbitals vs.~many-body symmetry) and are logically independent. A restricted reference does not by itself enforce SU$(2)$ adaptation of a correlation ansatz, whereas SU$(2)$ adaptation explicitly restricts the state to a chosen spin sector.
\end{remark}

\begin{remark}\label{re:kerSz}
We note that $\mathcal H_d^{\rm SU(2)}\neq\{0\}$ only when $d$ is even, e.g., $d=2k$. An exterior basis vector  $e_{I\uparrow}\wedge e_{J\downarrow}$ of $\mathcal{H}_d$ (with $|I|+|J|=d$) is an eigenvector of $S_{z}^{(d)}$. Here, each
spin-up (``\,$\uparrow$\,'') contributes $+\tfrac12$ to the eigenvalue and each spin-down (``\,$\downarrow$\,'') contributes $-\tfrac12$. Explicitly,
\begin{equation}
S_z^{(d)}\bigl(e_{I\uparrow}\wedge e_{J\downarrow}\bigr)
=
\frac12\bigl(|I|-|J|\bigr)\,e_{I\uparrow}\wedge e_{J\downarrow}.
\end{equation}
If $\psi\in \mathcal H_d^{\rm SU(2)}$, then in particular $S_z^{(d)}\psi=0$, so $\psi$ can have support only on basis vectors with $|I|=|J|$ enforcing $d=|I|+|J|=2|J|$. Therefore $\ker(S_z^{(d)})=\{0\}$ for odd $d$, and, consequently, $\mathcal H_d^{\rm SU(2)}\subseteq \ker(S_z^{(d)})$ implies $\mathcal H_d^{\rm SU(2)}=\{0\}$ whenever $d$ is odd.
\end{remark}

\begin{remark}[Proper embeddings of $\mathcal{H}_d^{{\rm SU}(2)}$]
\label{re:isspin}
Note that the condition $S_z\psi=0$ alone selects the sector $\ker(S_z)\subset\mathcal H_{d}$, isomorphic to
$\wedge^k\mathbb{C}^m\otimes \wedge^k\mathbb{C}^m$ via $e_{I\uparrow}\wedge e_{J\downarrow}\mapsto e_I\otimes e_J$.
Since $\ker(S_z)$ generally contains contributions from higher-spin irreducibles, we obtain a proper embedding 
\begin{equation}
\mathcal H_{d}^{\rm SU(2)}\subsetneq \ker(S_z)\cong \wedge^k\mathbb{C}^m\otimes \wedge^k\mathbb{C}^m.
\end{equation}
In particular, the conditions $S_{\pm}\psi = 0$ enforce the spin--up and spin--down orbitals to be isomorphic. 
Therefore, we instead consider the symmetric power $\operatorname{Sym}_2(\wedge^k\mathbb{C}^m)$, which has dimension
$\binom{\binom{m}{k}+1}{2}$. This space is isomorphic to the subspace of $\mathcal H_{d}$ generated by
\begin{equation}
e_{J\uparrow} \wedge e_{J\downarrow}, \qquad
e_{J\uparrow} \wedge e_{K\downarrow} + (-1)^{|J \backslash K|}e_{K\uparrow} \wedge e_{J\downarrow},
\end{equation}
where $J<K\in \binom{[n]}{k}$ are distinct subsets of size $k$. This symmetry, however, still does not enforce
SU(2)-invariance: the above generators need not vanish under the action of $S_\pm$.
For example, when $k=2$, consider the vector
\begin{equation}
\psi =
e_{1\uparrow}\wedge e_{2\uparrow}\wedge e_{3\downarrow}\wedge e_{4\downarrow}
+
e_{1\downarrow}\wedge e_{2\downarrow}\wedge e_{3\uparrow}\wedge e_{4\uparrow}.
\end{equation}
Here, $\psi \in \operatorname{Sym}_2(\wedge^k\mathbb{C}^m)$, but $S_+\psi \neq 0$. Hence $\mathcal H_{d}^{\rm SU(2)}$ is also a proper subspace of $\operatorname{Sym}_2(\wedge^k\mathbb{C}^m)$ and determining the invariant space itself requires more
refined tools. Nonetheless, the proper embedding of the invariant space $\mathcal{H}_d^{{\rm SU}(2)}$ into $\operatorname{Sym}_2(\wedge^k\mathbb{C}^m)$ will be useful in Section \ref{sec:cctheory}, as the larger space exhibits a simpler structure than the invariant space itself. 
\end{remark}

\medskip

We analyze the irreducible decomposition of the ${\rm SU}(2)$--representation $\mathcal{H}_d$ using characters. Recall that the \emph{character} of a finite-dimensional complex representation $V$ of ${\rm SU}(2)$ is the class function
\begin{equation}
    \chi_V:{\rm SU}(2)\to\mathbb{C}~;~ g\mapsto \operatorname{Tr}\!\left(g\big|_{V}\right).
\end{equation}
Here, $g|_V$ denotes the map $V \to V$ defined by a $2 \times 2$ matrix $g \in {\rm SU}(2)$, and $\operatorname{Tr}\!\left(g\big|_{V}\right)$ is the trace of the corresponding matrix of size $\dim(V)$.
The character $\chi_V$ is constant on conjugacy classes~\cite[Chapter~10]{michalek2021invitation}. Since every $U\in{\rm SU}(2)$ is unitary (hence normal), it is conjugate in ${\rm SU}(2)$ to an element of the maximal torus
\begin{equation}
\mathbb{T} := \left\{{\rm diag}(z,z^{-1})~:~|z|=1\right\}\cong S^1.
\end{equation}
Ergo, characters of ${\rm SU}(2)$ are determined by their restriction to $\mathbb{T}$. It is therefore convenient to write
$
\chi_V(z)\;:=\;\chi_V({\rm diag}(z,z^{-1})),
$
so that $\chi_V$ is a Laurent polynomial with finite support.
Moreover, since $t \in \mathbb{T}$ is conjugate with $t^{-1}$, the Weyl symmetry $\chi_V(z)=\chi_V(z^{-1})$ holds.

The irreducible representations of ${\rm SU}(2)$ are indexed by $j \in \tfrac{1}{2}\mathbb{Z}_{\ge 0}$ and may be realized as
$
V_j \cong {\rm Sym}^{2j}(\CC^2).
$
The action of $t={\rm diag}(z,z^{-1})\in \mathbb{T}$ on ${\rm Sym}^{2j}(\CC^2)$ is diagonal in the standard monomial basis, with eigenvalues $z^{2j},z^{2j-2},\ldots,z^{-2j}$. Hence, the irreducible character is
\begin{equation}
\label{eq:irrcharssu2}
\chi_j(z)
={\rm Tr}\left(t\big|_{V_j}\right)
=\sum_{i=0}^{2j} z^{2j-2i}
= z^{2j}+z^{2j-2}+\cdots+z^{-2j}.
\end{equation}
In particular, $\chi_{1/2}(z)=z+z^{-1}$. These characters are linearly independent and, by \cite[Proposition~2.1]{fulton2013representation}, an irreducible decomposition $V\cong \bigoplus_j V_j^{\oplus c_j}$ implies the character identity
\begin{equation}
\chi_V(z)=\sum_j c_j\,\chi_j(z).
\end{equation}
Thus the multiplicities $c_j$ (and in particular $c_0=\dim ( \mathcal H_d^{{\rm SU}(2)}$)) can be recovered by decomposing $\chi_V$ in the basis $\{\chi_j\}_j$ of irreducible characters.

\begin{lemma}
\label{thm:decompHd}
The character of the representation $\mathcal{H}_d$ of ${\rm SU}(2)$ is the Laurent polynomial
\begin{equation}\label{eq:charHd}
    \chi_{\mathcal{H}_d}(z)
    = \sum_{i = 0}^d \binom{m}{i}\binom{m}{d - i} z^{\,d-2i}.
\end{equation}
If $d = 2k$ then $\mathcal{H}_d$ decomposes into a sum of integer irreducible representations $V_j$
where $0 \le j \le k$, with multiplicity
\begin{equation}\label{eq:multofvj}
    c_{j} = \binom{m}{k - j}\binom{m}{k + j} - \binom{m}{k - j - 1}\binom{m}{k + j + 1}.
\end{equation}
If $d = 2k + 1$ then $\mathcal{H}_d$ decomposes into a sum of half-integer irreducible representations
$V_{j + \frac{1}{2}}$ where $0 \le j \le k$ with multiplicity 
\begin{equation}\label{eq:repmultodd}
c_{j + \frac{1}{2}}
=
\binom{m}{k-j}\binom{m}{k+j+1}
-
\binom{m}{k-j-1}\binom{m}{k+j+2}.
\end{equation}
\end{lemma}

\begin{proof}
Note that on the one-particle space $\mathcal H\cong \CC^m\otimes\CC^2$, the element ${\rm diag}(z,z^{-1})\in \mathbb{T}$ acts as the
$n\times n$ diagonal matrix
\begin{equation}
g_n(z)={\rm diag}(z,z^{-1},z,z^{-1},\dots,z,z^{-1}).
\end{equation}
Hence, the character of $\mathcal H_d=\wedge^d\mathcal H$ is the trace of the $d$th exterior power of $g_n(z)$.
Its diagonal entries are the principal $d\times d$ minors of $g_n(z)$, each obtained by selecting $d$ diagonal
entries of $g_n(z)$. Choosing $i$ entries equal to $z^{-1}$ and $d-i$ entries equal to $z$ yields the monomial
$z^{d-i}(z^{-1})^{i}=z^{d-2i}$, and the number of such choices is $\binom{m}{i}\binom{m}{d-i}$, proving Eq.~\eqref{eq:charHd}.

\medskip
We now decompose $\chi_{\mathcal H_d}$ into the irreducible characters in Eq.~\eqref{eq:irrcharssu2}.
We first assume ${d=2k}$, and without loss of generality, we may assume $2k \le m$.  We determine the multiplicities recursively by matching the highest power of $z$.
Write
\begin{equation}
\chi_{\mathcal H_d}(z)=\sum_{j\in\mathbb Z} a_j\,z^{2j}, \qquad a_j=\binom{m}{k-j}\binom{m}{k+j},
\qquad a_{-j}=a_j.
\end{equation}
with the convention $\binom{m}{r}=0$ for $r\notin\{0,\dots,m\}$. In particular, $a_j=0$ for $j>k$, so the
highest power occurring in $\chi_{\mathcal H_d}(z)$ is $z^{2k}$.
Its coefficient is $a_{k}$, so we set $c_{k}:=a_{k}$ and consider
\begin{equation}
\chi^{(1)}(z):=\chi_{\mathcal H_d}(z)-c_{k}\chi_{k}(z).
\end{equation}
Since $\chi_{\mathcal H_d}$ is invariant under $z\mapsto z^{-1}$ and has nonnegative integer coefficients, the same holds
for $\chi^{(1)}$. Moreover, the terms $z^{2k}$ and $z^{-2k}$ are removed, and the remaining highest term is $z^{2(k-1)}$ with coefficient $c_{k - 1}$.
Continuing downward, we find that the coefficient of $z^{2j}$ in $\chi_{\mathcal H_d}(z)$ is $\sum_{\ell\ge j} c_\ell$,
so the multiplicities satisfy
\begin{equation}
c_j=a_j-a_{j+1}
=
\binom{m}{k-j}\binom{m}{k+j}-\binom{m}{k-j-1}\binom{m}{k+j+1}.
\end{equation}
Now, let $d=2k+1$. The character formula \eqref{eq:charHd} becomes
\begin{equation}
\chi_{\mathcal H_{2k+1}}(z)=\sum_{j \in \ZZ}\binom{m}{k - j }\binom{m}{k+1 + j}\,z^{\,2j + 1},
\end{equation}
hence $\chi_{\mathcal H_{2k+1}}(z)$ is a Laurent polynomial with only odd exponents. Therefore, in the decomposition of
$\mathcal H_{2k+1}$ into irreducibles, only half-integer representations can occur.
We determine the coefficients $c_{j + \frac12}$ by the same triangular elimination as in the even case, matching the highest powers of $z$ successively. This gives the formula in Eq. (\ref{eq:repmultodd}). 
\end{proof}

\begin{theorem}
\label{th:DimHSU2}
For $d=2k$ electrons in $n = 2m$ spin orbitals, the dimension of the spin singlet sector $\mathcal{H}_d^{{\rm SU}(2)} \subsetneq \mathcal{H}_d$ is given by
\begin{equation}
\dim\bigl(\mathcal{H}_{d}^{{\rm SU}(2)}\bigr)
= N(m + 1, k + 1),
\end{equation}
where $N(m+1,k+1)$ is a Narayana number (see \texttt{A001263}).
\end{theorem}

\begin{proof}
This is a direct consequence of  Lemma~\ref{thm:decompHd} since
\begin{small}
    $$
\dim\bigl(\mathcal{H}_{d}^{{\rm SU}(2)}\bigr) = c_0 = \binom{m}{k}^2 - \binom{m}{k - 1}\binom{m}{k + 1} = \frac{1}{m + 1}\binom{m + 1}{k + 1}\binom{m + 1}{k}
= N(m + 1, k + 1).
$$
\end{small}
\end{proof}

\begin{remark}
The Narayana numbers form a triangle of positive integers for $0\le k\le m$ and refine the Catalan numbers via
\begin{equation}
\label{eq:catsum}
\sum_{r = 0}^m N(m + 1,r + 1) = C_{m + 1}.
\end{equation}
The Catalan number $C_{m+1}$ counts Dyck paths from $(0,0)$ to $(2m+2,0)$, while the Narayana number $N(m+1,k+1)$ counts Dyck paths with exactly $k+1$ peaks (equivalently, $k$ pits).
\end{remark}

\begin{example}[$2k=d=4;\,2m = n = 8$]
We consider $4$ electrons in $4$ spatial orbitals. Then $\mathcal{H}_4\cong \wedge^4\mathbb{C}^8$ has
dimension $\dim (\mathcal{H}_4) = \binom{8}{4}=70$ and Theorem~\ref{th:DimHSU2} yields 
\begin{equation}
\dim\bigl(\mathcal{H}_4^{{\rm SU}(2)}\bigr)=N(5,3)=20.
\end{equation}
The invariant space
$\mathcal{H}_4^{{\rm SU}(2)}$ is the intersection of the kernels of the three $70 \times 70$ matrices $S_z^{(d)},S_+^{(d)},S_-^{(d)}$ and it is generated by the following $20$ vectors:
\begin{align*}
 &e_{1\uparrow,\,1\downarrow,\,2\uparrow,\,2\downarrow}, \,\, e_{1\uparrow,\,1\downarrow,\,2\downarrow,\,3\uparrow}-e_{1\uparrow,\,1\downarrow,\,2\uparrow,\,3\downarrow}, \,\, e_{1\downarrow,\,2\uparrow,\,2\downarrow,\,3\uparrow}-e_{1\uparrow,\,2\uparrow,\,2\downarrow,\,3\downarrow}, \,\,
e_{1\uparrow,\,1\downarrow,\,3\uparrow,\,3\downarrow},\\
 &e_{1\downarrow,\,2\uparrow,\,3\uparrow,\,3\downarrow}-e_{1\uparrow,\,2\downarrow,\,3\uparrow,\,3\downarrow}, \,\, e_{2\uparrow,\,2\downarrow,\,3\uparrow,\,3\downarrow}, \,\, e_{1\uparrow,\,1\downarrow,\,2\downarrow,\,4\uparrow}
-e_{1\uparrow,\,1\downarrow,\,2\uparrow,\,4\downarrow}, \,\, e_{1\downarrow,\,2\uparrow,\,2\downarrow,\,4\uparrow}-e_{1\uparrow,\,2\uparrow,\,2\downarrow,\,4\downarrow},\\
 &e_{1\uparrow,\,1\downarrow,\,3\downarrow,\,4\uparrow}-e_{1\uparrow,\,1\downarrow,\,3\uparrow,\,4\downarrow}, \,\, e_{1\downarrow,\,2\uparrow,\,3\downarrow,\,4\uparrow}-e_{1\uparrow,\,2\downarrow,\,3\downarrow,\,4\uparrow}-e_{1\downarrow,\,2\uparrow,\,3\uparrow,\,4\downarrow}+e_{1\uparrow,\,2\downarrow,\,3\uparrow,\,4\downarrow},\\
 &e_{2\uparrow,\,2\downarrow,\,3\downarrow,\,4\uparrow} - e_{2\uparrow,\,2\downarrow,\,3\uparrow,\,4\downarrow}, \,\, e_{1\downarrow,\,2\downarrow,\,3\uparrow,\,4\uparrow}-e_{1\uparrow,\,2\downarrow,\,3\downarrow,\,4\uparrow}-e_{1\downarrow,\,2\uparrow,\,3\uparrow,\,4\downarrow}+e_{1\uparrow,\,2\uparrow,\,3\downarrow,\,4\downarrow},\\
 &e_{1\downarrow,\,3\uparrow,\,3\downarrow,\,4\uparrow}-e_{1\uparrow,\,3\uparrow,\,3\downarrow,\,4\downarrow}, \,\,
e_{2\downarrow,\,3\uparrow,\,3\downarrow,\,4\uparrow} - e_{2\uparrow,\,3\uparrow,\,3\downarrow,\,4\downarrow},
\,\, e_{1\uparrow,\,1\downarrow,\,4\uparrow,\,4\downarrow}, \,\,
e_{1\downarrow,\,2\uparrow,\,4\uparrow,\,4\downarrow} - e_{1\uparrow,\,2\downarrow,\,4\uparrow,\,4\downarrow}, \\
 &e_{2\uparrow,\,2\downarrow,\,4\uparrow,\,4\downarrow}, \,\, e_{1\downarrow,\,3\uparrow,\,4\uparrow,\,4\downarrow}-e_{1\uparrow,\,3\downarrow,\,4\uparrow,\,4\downarrow}, \,\,
e_{2\downarrow,\,3\uparrow,\,4\uparrow,\,4\downarrow}-e_{2\uparrow,\,3\downarrow,\,4\uparrow,\,4\downarrow}, \,\, e_{3\uparrow,\,3\downarrow,\,4\uparrow,\,4\downarrow}.
\end{align*}
\end{example}

\subsection{Spin of Fock States and Spin-Singlet Fock States}

We now pass from the $d$-electron sector to the full fermionic Fock space i.e., the direct sum of all particle-number sectors (see Section~\ref{sec:FermiDiracAlgebra}). The induced $\mathfrak{su}(2)$--action preserves each particle sector $\mathcal H_k=\wedge^k\mathcal H$ and, in the canonical exterior basis, the action of element $g \in \mathfrak{su}(2)$ is represented by the block diagonal matrix
\begin{equation}\label{eq:actmatF}
    \hat{g}= g^{(0)} \oplus g^{(1)} \oplus \cdots \oplus g^{(2m)},
\end{equation}
where $g^{(\ell)}$ is defined in Eq.~\eqref{eq:wedgeaction_insertion}.
We can now define the subspace $\mathcal{F}^{{\rm SU}(2)} \subsetneq \mathcal{F}$ of total-spin singlet states as the direct sum of total--spin singlets in all particle sectors, i.e., 
\begin{equation}
\mathcal F=\bigoplus_{\ell=0}^{2m}\mathcal H_\ell,
\quad {\rm and} \quad 
\mathcal F^{\mathrm{SU}(2)}=\bigoplus_{\ell=0}^{2m}\mathcal H_\ell^{\mathrm{SU}(2)}
=\bigoplus_{k=0}^{m}\mathcal H_{2k}^{\mathrm{SU}(2)}.
\end{equation}
Since $\mathcal H_\ell^{{\rm SU}(2)}=\{0\}$ for odd $\ell$, while $\dim(\mathcal H_{2k}^{{\rm SU}(2)})=N(m+1,k+1)$ for $0\le k\le m$, the identity in Eq.~\eqref{eq:catsum} yields
\begin{equation}
\dim\bigl(\mathcal{F}^{{\rm SU}(2)}\bigr)
=
\sum_{k=0}^{m} N(m+1,k+1)
=
C_{m+1}.
\end{equation}
We recall that ${\rm FD}_{2m}\cong {\rm End}(\mathcal F)\cong \mathcal F\otimes \mathcal F^\dagger$, hence, we can write elements FD$_{2m}$ as
\begin{equation}
\hat \Omega = \sum_{I,J \subseteq [\![m]\!]} a_{I,J}\, e_I \otimes e_J^\dagger .
\end{equation}
Therefore, we can canonically induce an $\mathfrak{su}(2)$--action on ${\rm End}(\mathcal F)$, via the action on $\mathcal F$ and its conjugate $\mathcal F^\dagger$. This yields the following action for $\mathfrak{su}(2)$,
\begin{equation}
\label{eq:FDaction}
g(\hat \Omega)
=
\sum_{I,J \subseteq [\![m]\!]} a_{I,J}\, g(e_I)\otimes e_J^\dagger
-
\sum_{I,J \subseteq [\![m]\!]} a_{I,J}\, e_I\otimes g(e_J^\dagger)
=
\hat g \hat \Omega - \hat \Omega \hat g
=
[\hat g,\hat \Omega].
\end{equation}
Here $\hat g$ is the $2^{2m}\times 2^{2m}$ matrix defined in Eq.~\eqref{eq:actmatF}. 
The operator $\hat g$ is an endomorphism of $\mathcal{F}$ and thus it is an element in the Fermi--Dirac algebra.
The Fermi--Dirac representation of the generators $S_{\pm}$ and $S_z$ of $\mathfrak{su}(2)$ is
\begin{equation}
\label{eq:su2FDgens}
    \hat S_{+} = -\sum_{k = 1}^m a_{k\uparrow}^\dag a_{k \downarrow}, \quad
    \hat S_- = -\sum_{k = 1}^m a_{k\downarrow}^\dag a_{k\uparrow}, \quad
    \hat S_z = \frac{1}{2}\sum_{k = 1}^m (a_{k\uparrow}^\dag a_{k\uparrow} -a_{k\downarrow}^\dag a_{k\downarrow}).
\end{equation}
Note that the elements in the Fermi--Dirac algebra, ${\rm FD}_{2m}$, that vanish under the $\mathfrak{su}(2)$--action defined in Eq.~\eqref{eq:FDaction} are precisely those that commute with
$\hat S_\pm$ and $\hat S_z$.

\section{Spin-Adapted Coupled Cluster Theory}\label{sec:cctheory}

The (spin-free) electronic Hamiltonian in Eq.~\eqref{eq:Hsecondquant} is $\operatorname{SU}(2)$-invariant. This follows since it commutes with the generators of the induced $\mathfrak{su}(2)$--action on Fock space. In symbols
\begin{equation}\label{eq:HcommutesSU2}
[\hat H,\hat S_+] = [\hat H,\hat S_-] = [\hat H,\hat S_z] = 0.
\end{equation}
Consequently, $\hat H$ commutes with the Casimir operator
$\hat S^2$, and preserves the decomposition of the Hilbert space into $\operatorname{SU}(2)$-isotypic components. 
The \textit{spin-singlet Schrödinger equation} can therefore be defined as the eigenvalue problem for $\operatorname{SU}(2)$-invariant quantum states,
\begin{equation}\label{eq:schr}
    \hat H\psi = E\psi, \qquad \psi \in \mathcal{H}_d^{\operatorname{SU}(2)}.
\end{equation}

We recall that coupled cluster theory employs the ansatz 
\begin{equation}
\psi = e^{\hat T} e_{[\![k]\!]}
\end{equation}
where {\it the cluster matrix} $\hat T \in {\rm End}(\mathcal{F})$ is the new unknown and $e_{[\![k]\!]}\in\mathcal{F}$ is the reference state~\cite{coester1958bound,vcivzek1966correlation}. 
In particular, we choose the reference state as a closed--shell spin--restricted state  (as in restricted Hartree--Fock theory).

\begin{remark}\label{re:refstate}
    The \textit{closed-shell spin-restricted reference state} 
    $$
    \phi_0 = e_{[\![k]\!]} = e_{1 \downarrow} \wedge e_{1 \uparrow} \wedge \cdots \wedge e_{k \downarrow} \wedge e_{k \uparrow}
    $$ 
    is a singlet state. This can be seen by observing that it is annihilated by the generators $\hat S_{\pm}$ and $\hat S_{z}$. Note that the $\hat S_z$--action vanishes by Remark~\ref{re:kerSz}. Moreover, the raising operator $\hat S_+$ acts by replacing a $\downarrow$-spin electron in a given orbital by an $\uparrow$-spin electron in the same orbital. Since the reference state is a closed shell state, it follows by Pauli's exclusion principle that $\hat S_{+}e_{[\![k]\!]} = 0$. The same argument holds for $\hat S_-$.
\end{remark}

Since the reference state $e_{[\![k]\!]}$ is a singlet, the physically relevant correlated ansätze may be taken entirely within the singlet sector. In other words, we can formulate the coupled cluster equations within the ${\rm SU}(2)$--invariant space of $\mathcal{H}_{d}$.

\subsection{Spin Singlet Truncation Varieties}

Coupled cluster theory is commonly formulated in the molecular orbital (particle--hole) picture~\cite{vcivzek1966correlation}.
Accordingly, we focus on the subalgebra of ${\rm FD}_{2m}$ generated by the occupied annihilators $\hat a_{i\alpha}$, $i\in[k]$,
and the virtual creators $a^\dagger_{b\alpha}$, $b\in[m]\setminus[k]$ (with $\alpha\in\{\uparrow,\downarrow\}$).
These generators mutually anticommute, so this subalgebra is an exterior algebra. In symbols,
\[
\mathbb{C}\langle
\hat a_{1\downarrow}, \hat a_{1\uparrow}, \dots, \hat a_{k\downarrow}, \hat a_{k\uparrow},
\hat a_{k+1,\downarrow}^\dagger, \hat a_{k+1,\uparrow}^\dagger, \dots, \hat a_{m\downarrow}^\dagger, \hat a_{m\uparrow}^\dagger
\rangle/\langle G\rangle
\;\cong\;
\wedge \mathbb{C}^{2m},
\]
where $G$ denotes the set of relations in Eq.~\eqref{eq:fermirels}.
Since the electronic Hamiltonian conserves particle number, we restrict to the particle--number preserving subalgebra $\mathcal{V}$ generated by the excitation
words $\hat a^\dagger_{b\alpha}\hat a_{i\beta}$ of bidegree $(1,1)$, where $i\in[k]$, $b\in[m]\setminus[k]$, and $\alpha,\beta\in\{\uparrow,\downarrow\}$.
Since these generators are words of degree two in FD$_{2m}$, the induced product on $\mathcal V$ is commutative,
and hence $\mathcal V$ is a commutative ring. It is also a graded ring graded by the excitation levels. Hence $\mathcal{V} = \bigoplus_{\ell = 0}^{d} \mathcal{V}^{(\ell)}$, where the $\ell$th grading $\mathcal{V}^{(\ell)}$ is a vector space generated by words with excitation level $\ell$ (degree $\ell$ monomials):
$$
\hat a_{b_1, \alpha_1}^\dag \hat a_{i_1 \beta_1} \cdots \hat a_{b_\ell, \alpha_\ell}^\dag \hat a_{i_\ell, \beta_\ell} = \hat a_{b_\ell, \alpha_\ell}^\dag  \cdots \hat a_{b_1, \alpha_1}^\dag  \hat a_{i_1 \beta_1} \cdots \hat a_{i_\ell \beta_\ell}.
$$
By Pauli's exclusion principle (${a_{p}^\dag}^2 = 0$) it follows that $\mathcal{V}^{(d + 1)} = \{0\}$. That is, the elements of $\mathcal{V}$ are nilpotent of order at most $d$. Elements in $\mathcal{V}$ can thus be written as sums of graded elements: $\hat T =\hat T_0 +  \hat T_1 + \hat T_2 + \cdots + \hat T_{d}$ where
$$
\hat T_\ell = \sum_{\substack{I \subseteq [\![k]\!],B \subseteq[\![m]\!] \backslash [\![k]\!] \\ |I| = |B| = \ell}} t_{I,B} \hat a_{b_\ell, \alpha_\ell}^\dag  \cdots \hat a_{b_1, \alpha_1}^\dag  \hat a_{i_1 \beta_1} \cdots \hat a_{i_\ell \beta_\ell}  = \sum_{\substack{I \subseteq [\![k]\!],B \subseteq[\![m]\!] \backslash [\![k]\!] \\ |I| = |B| = \ell}} t_{I,B}\hat a_B^\dagger \hat a_I \,\,\in \,\, \mathcal{V}^{(\ell)}.
$$
The coordinates $t_{I,B} \in \CC$ of $\mathcal{V}$ are called \textit{cluster amplitudes}.
Note that the indexing sets $[\![m]\!]$ and $[\![k]\!]$ are of size $n = 2m$ and $d = 2k$ respectively and hence the combinatorial identity ${\sum_{\ell = 0}^d \binom{d}{\ell}\binom{n - d}{\ell} = \binom{n}{d}}$ reveals that $\mathcal{V} \cong \CC^{\binom{n}{d}}$.

\begin{proposition}
\label{prop:HdVd_isomorphic}
With the induced $\mathfrak{su}(2)$--actions on $\mathcal H_{2k}\subseteq\mathcal F$ and on
$\mathcal V$ (via $g\cdot \hat T=[\hat g,\hat T]$, cf.~Eq.~\eqref{eq:FDaction}), the map
\begin{equation}
    \Phi: \mathcal{V} \to \mathcal{H}_{2k}, \quad \hat T \mapsto \hat T e_{[\![k]\!]}.
\end{equation}
is an isomorphism of $\mathfrak{su}(2)$--modules.
\end{proposition}

\begin{proof}
The coordinates of $\Phi(\hat T)$ are of the form $\hat a_B^\dagger \hat a_I\, e_{[\![k]\!]} = (-1)^{\sigma}e_J$, where $J =( [\![k]\!] \backslash I) \cup B$ and the sign $\sigma$ is the sign of the permutation $[\![k]\!]\mapsto (I, J\cap [\![k]\!])$. This proves $\Phi$ is a bijection.
The reference state $e_{[\![k]\!]}$ is a closed-shell spin singlet by Remark \ref{re:refstate}. Therefore, for any $\hat T\in\mathcal V$ and any $g \in \mathfrak{su}(2)$ we get that
\begin{equation}
\Phi(g\cdot \hat T)
= \Phi([\hat g,\hat T])
= [\hat g,\hat T]\,e_{[\![k]\!]}
= \hat g\,\hat T e_{[\![k]\!]} - \hat T\,\hat g\,e_{[\![k]\!]}
= \hat g\,\Phi(\hat T)
= g\cdot \Phi(\hat T).
\end{equation}
Hence $\Phi$ is $\mathfrak{su}(2)$--equivariant, and therefore an isomorphism of $\mathfrak{su}(2)$--modules.
\end{proof}

The isomorphism above also induces a grading on $\mathcal{H}_{2k}$: we say a basis vector $e_{J}$ of $\mathcal{H}_{2k}$ has degree (or excitation level) $|J \backslash [\![k]\!]|$, and we denote the $\ell$th grading of $\mathcal{H}_{2k}$ by $\mathcal{H}_{2k}^{(\ell)}$. 
As a consequence of Proposition \ref{prop:HdVd_isomorphic}, we can study the irreducible decomposition of $\mathcal H_d$ via the irreducible decomposition of $\mathcal V$. In particular:

\begin{corollary}\label{cor:invariants}
    The map $\Phi$ restricts to the invariant spaces of $\mathcal{V}$ and $\mathcal{H}_{2k}$. Therefore
    \[
\mathcal H_{2k}^{\operatorname{SU}(2)} \cong \mathcal V^{\operatorname{SU}(2)}.
\]
\end{corollary}

We want to find a homogeneous basis for $\mathcal V^{\operatorname{SU}(2)}$. This is possible since the generators in Eq.~\eqref{eq:su2FDgens} are homogeneous, and hence the $\mathfrak{su}(2)$--action preserves excitation levels. In particular, each basis element lives in a grading $\mathcal{V}^{(\ell)}$ for some $1 \le \ell \le m$.
We start with the simplest generators of $\mathcal V^{{\rm SU}(2)}$, namely those of excitation level $1$. The {\it excitation operators} are the homogeneous elements of bidegree $(1,1)$ given by:
\[
\hat X_{i,b} = \hat a_{b\uparrow}^\dagger \hat a_{i\uparrow} + \hat a_{b\downarrow}^\dagger \hat a_{i\downarrow},
\qquad 1\le i\le k < b\le n.
\]
They are linearly independent and commute with $\hat S_\pm$ and $\hat S_z$ and are therefore generators of $\mathcal V^{\operatorname{SU}(2)}$.
The \textit{excitation ring} is the graded ring
\[
R=\mathbb{C}[\hat X_{1,k+1}, \hat X_{1,k+2}, \ldots, \hat X_{1,n},\hat X_{2,k+1}, \ldots  \hat X_{k,n}] \subseteq \mathcal{V},
\]
i.e., $R$ is the subring of $\mathcal V$ generated by the $k(m - k)$ excitation operators. Since $R$ is generated by degree one elements in $\mathcal{V}$, it inherits the grading of $\mathcal{V}$ and $R^{(\ell)} \subseteq \mathcal{V}^{(\ell)}$ for all $1 \le \ell \le m$.
Since the variables $\hat X_{i,b}$ commute with $\hat S_\pm$ and $\hat S_z$, every polynomial in $R$ also commutes with them and therefore we obtain a stricter inclusion $R \subseteq \mathcal V^{{\rm SU}(2)}$.

\begin{theorem}\label{thm:exring}
The invariant space $\mathcal V^{\operatorname{SU}(2)}$ is an Artinian ring. The excitation ring $R$ coincides with the invariant ring, that is $R=\mathcal V^{\operatorname{SU}(2)}$.
\end{theorem}

\begin{proof}
The ring $\mathcal{V}$ is Artinian, and hence its invariant ring $\mathcal V^{\operatorname{SU}(2)}$ is too.
We notice that $R$ is a subring of the ring $S$ generated by level-one words $\hat x_{i,b, \alpha} := \hat a_{b\alpha}^\dag \hat a_{i\alpha}$ where
$1\le i\le k<b\le m$ and $\alpha\in\{\downarrow,\uparrow\}$. Note $S$ is a proper subring of $\mathcal{V}$, since we restrict to products with a fixed spin. The ring $S$ has $2k(m-k)$ variables and can implicitly be described as
\[
S \cong \mathbb{C}[x_{1,k+1,\downarrow}, x_{1,k+1,\uparrow}, \dots, x_{k,m,\downarrow}, x_{k,m,\uparrow}]
\Big/
\langle x_{i,a,\alpha}x_{j,b,\alpha} + x_{j,a,\alpha}x_{i,b,\alpha}\rangle .
\]
The ring $R$ is isomorphic to the subring of $S$ generated by the $k(m-k)$ linear combinations
$\hat x_{i,b} := \hat x_{i,b,\uparrow}+\hat x_{i,b,\downarrow} = \hat X_{i,b}$, where $1 \le i \le k \le b \le m$. They fulfill the degree-three relations
\begin{equation}\label{eq:gens}
\sum_{\tau\in S_3} \hat x_{i,\tau(a)} \hat x_{j,\tau(b)} \hat x_{\ell,\tau(c)}=0,
\quad\text{where } i\le j\le \ell \text{ and } a\le b\le c .
\end{equation}
Let $I$ be the ideal generated by these relations and so $R\cong \mathbb{C}[x_{1,k+1},\dots,x_{k,m}]/I$.
In \cite[Theorem~1.2]{price2026plane}, Sverrisdóttir, together with Price and Stelzer, prove that these relations form a Gr\"obner basis and that the number of
standard monomials is $N(m+1,k+1)$. Therefore $R$ is a vector space of dimension $N(m+1,k+1)$. On the other hand, by Corollary~\ref{cor:invariants} and Theorem~\ref{th:DimHSU2}:
\begin{equation}
\dim \mathcal V^{{\rm SU}(2)} = \dim \mathcal H_{2k}^{{\rm SU}(2)} = N(m+1,k+1).
\end{equation}
Hence, it follows the claim $R=\mathcal V^{{\rm SU}(2)}$.
\end{proof}

\begin{remark}[Explicit generators of the invariance space]
The standard monomials of $R$ mentioned in the proof of Theorem \ref{thm:exring} form the desired homogeneous basis for $\mathcal{V}^{{\rm SU}(2)}$.
We look at the $2$nd grading of $R$, i.e., the space $R^{(2)}$ of homogeneous polynomials in $R$ of degree $2$.
There are no linear relations among the $\binom{km-k^2+1}{2}$ degree two monomials. Therefore they are standard monomials of $R$ and form a basis for $R^{(2)}$. Together with the unit $1$ and the degree one generators $\hat X_{i,b}$, they yield a basis for the subspace of
$\mathcal V^{{\rm SU}(2)}$ consisting of elements of degree at most $2$. 
Now consider the $3$rd grading $R^{(3)}$. There are linear relations among the degree three monomials in $\hat X_{i,b}$, namely the generators
in \eqref{eq:gens}. Consequently, an explicit description of the standard monomials and hence the basis elements of excitation level $\ge 3$ requires additional combinatorial tools; see
\cite{price2026plane}. Since we focus here on the CCD and CCSD formalisms, for which only invariants up to degree two
are required, we do not pursue these higher-degree invariants further.
\end{remark}

For $1 \le \ell \le k$ we define the \textit{$\ell$th truncated cluster operator} $\hat T_\ell$ to be a general element in the $\ell$th grading of $\mathcal{V}^{\operatorname{SU}(2)}$.
We can write the truncated cluster operators $\hat T_1$ and $\hat T_2$ explicitly:
\begin{equation}\label{eq:clustops}
    \hat T_1 = \sum_{(i,b)} t_{i,b}\hat X_{i,b}, \quad \hat T_2 = \sum_{(i,b) \le (j,c)} t_{ij,bc}\hat X_{i,b}\hat X_{j,c}, \quad\text{where } t_{i,b}, t_{ij,bc} \in \CC.
\end{equation}
Here the pairs $(i,b) \in [k] \times [m] \backslash [k]$ are ordered lexicographically, and the condition $(i,b)\le (j,c)$ avoids double counting of quadratic monomials. We then define the \textit{cluster operator} $\hat T$ as the sum of the truncated cluster operators, that is $\hat T = \hat T_1 + \hat T_2 + \cdots + \hat T_k$. We notice that $\hat T$ is a general element in the invariant space $\mathcal{V}^{{\rm SU}(2)}$ with a zero constant term.
We denote the space of cluster operators as 
$$
{\mathcal{V}^{{\rm SU}(2)}}' = \bigoplus_{\ell = 1}^{2k}  \mathcal{V}^{{\rm SU}(2)}_{(\ell)},
\quad {\rm with} \quad 
\mathcal V^{{\rm SU}(2)}_{(\ell)} := \mathcal V^{{\rm SU}(2)}\cap \mathcal V^{(\ell)}.
$$
This is the space of ${\rm SU}(2)$--invariant operators in $\mathcal{V}$ with a zero constant term. Equivalently this is the space of nilpotent operators in $\mathcal{V}^{{\rm SU}(2)}$, that is those $\hat T$ such that $\hat T^{2k + 1} = 0$.

\begin{proposition}\label{prop:spin_adapted_exp}
    The \textit{spin-singlet exponential map}
    \begin{equation}\label{eq:expmap}
{\mathcal V^{{\rm SU}(2)}}' \to {\mathcal H_{2k}^{{\rm SU}(2)}},
\qquad
\hat T \mapsto \exp(\hat T)\,e_{[\![k]\!]},
\end{equation}
    is an injective polynomial map.
    Its image coincides with the subspace
    \[
    \{ \psi \in {\mathcal{H}_{2k}^{{\rm SU}(2)}} ~:~ e_{[\![k]\!]}^T \psi = 1 \}\subset {\mathcal{H}_{2k}^{{\rm SU}(2)}}.
    \]
    and the corresponding restriction is a bijection with a polynomial inverse.
\end{proposition}

\begin{proof}
    Take a cluster operator $\hat T  \in {V^{{\rm SU}(2)}}'$. We recall that $\hat T$ is nilpotent and
    therefore $\exp{(\hat T)}e_{[\![k]\!]}$ has polynomial entries. The operator $\exp(\hat T)$ also commutes with the generators $\hat S_{\pm}$ and $\hat S_z$ so $\exp(\hat T)e_{[\![k]\!]} \in \mathcal{H}_{2k}^{{\rm SU}(2)}$ and the map is well defined and polynomial. Since $\hat T$ is nilpotent, the diagonal entries of $\exp{(\hat T)}$ are $1$ and so $e_{[\![k]\!]}^T \exp{(\hat T)}e_{[\![k]\!]} = 1$.
    The bijectivity and the existence of a polynomial inverse follow from the same excitation-level induction as in~\cite[Proposition~2.4]{faulstich2024algebraic}, applied within the invariant subspaces.
\end{proof}

Fix a \textit{level set} $\sigma \subseteq [d]$, specifying the allowed excitation levels. For instance, when employing CCS, CCD or CCSD, we set $\sigma= \{1\}, \{2\}, \{1,2\}$, respectively. Let $\mathcal V_\sigma^{{\rm SU}(2)}$ denote the subspace of $\mathcal{V}^{{\rm SU}(2)}$ consisting of operators whose excitation level lies in $\sigma$. Explicitly:
\begin{equation}
\mathcal V_\sigma^{{\rm SU}(2)}
:=
\bigoplus_{\ell\in\sigma} \mathcal{V}^{{\rm SU}(2)}_{(\ell)} \subseteq {\mathcal{V}^{{\rm SU}(2)}}'.
\end{equation}
The restriction of the spin-singlet exponential map, defined in Proposition \ref{prop:spin_adapted_exp}, to this subspace is injective. 
We may further compose it with the projective embedding
\[
\mathcal H_{2k}^{{\rm SU}(2)} \hookrightarrow \PP \mathcal H_{2k}^{{\rm SU}(2)}
\cong \PP^{N(m+1,k+1)-1}.
\]
Analogous to the truncation varieties in \cite{sverrisdottir2025algebraic,faulstich2024algebraic} we define the \textit{spin-singlet truncation variety} $V_\sigma^{\operatorname{SU}(2)} \subseteq \PP\mathcal{H}_{2k}^{\operatorname{SU}(2)}$ to be the Zariski closure of the image of $\mathcal{V}_{\sigma}^{{\rm SU}(2)}$ under this composed map. Hence, the spin-singlet truncation varieties are parameterized by restrictions of the spin-singlet exponential map. 
In this paper we focus on CCS, CCD and CCSD, where the elements of the subspaces $\mathcal{V}_{\sigma}^{\operatorname{SU}(2)}$ are the truncated cluster operators $\hat T_1$, $\hat T_2$ and $\hat T_1 + \hat T_2$ in Eq. \eqref{eq:clustops}. The corresponding truncation varieties are then parameterized by vectors:
 $$
 \psi = \exp(\hat T_1)e_{[\![k]\!]}, \qquad \psi = \exp(\hat T_2)e_{[\![k]\!]}, \qquad \psi = \exp(\hat T_1 + \hat T_2)e_{[\![k]\!]}.
 $$

\begin{example}[CCS for $m = 4$, $k = 2$]
    The CCS cluster operator is of the form
    $$
    \hat T_1 = t_{1,3}\hat X_{1,3} + t_{2,3}\hat X_{2,3} + t_{1,4}\hat X_{1,4} + t_{2,4}\hat X_{2,4}.
    $$
    We recall from \cite[Theorem 3.5]{faulstich2024algebraic} that the (spin-orbital) CCS truncation variety is the Grassmannian $\operatorname{Gr}(4,8)\subseteq \PP^{69}$ in its Pl\"ucker embedding.
    To relate this to the spin-singlet truncation variety, we compose the
spin-singlet CCS exponential map with the natural inclusion $\mathcal H_{4}^{{\rm SU}(2)}\hookrightarrow \mathcal H_{4}\cong  \CC^{70}$.
This yields a parametrization $\CC^{4}\to \PP \mathcal H_{4} \cong \PP^{69}$ of $V_{\{1\}}^{{\rm SU}(2)}$ in the cluster amplitudes $(t_{1,3},t_{2,3},t_{1,4},t_{2,4})$,
given by the maximal minors of the $4 \times 8$ matrix
    \begin{equation}\label{eq:grparam}
        \begin{pmatrix}
        1 & 0 & 0 & 0 & t_{1,3} & 0 & t_{1,4} & 0\\
        0 & 1 & 0 & 0 & 0 & t_{1,3} & 0 & t_{1,4}\\
        0 & 0 & 1 & 0 & t_{2,3} & 0 & t_{2,4} & 0\\
        0 & 0 & 0 & 1 & 0 & t_{2,3} & 0 & t_{2,4}\\
    \end{pmatrix}.
    \end{equation}
    Here the rows are indexed by $[\![2]\!]$ and the columns by $[\![4]\!]$. We further project the image onto the $20$ dimensional subspace $\PP\operatorname{Sym}_2(\wedge^2 \CC^4)$ described in Remark \ref{re:isspin}. 
   The resulting map is the quadratic Veronese embedding  $\nu_2$ of the Pl\"ucker embedding of $\operatorname{Gr}(2,4)\subseteq \mathbb P^{5}$.
    More precisely, on the standard affine chart $p_{12}=1$, the Pl\"ucker coordinates are
    $$
    (p_{12},p_{13},p_{23},p_{14},p_{24},p_{34})
    =
    \bigl(1,\ t_{2,3},\ -t_{1,3},\ t_{2,4},\ -t_{1,4},\ t_{1,3}t_{2,4}-t_{2,3}t_{1,4}\bigr).
    $$
Composing with $\nu_2$ gives all degree two monomials in the Pl\"ucker coordinates.
    Consequently, the spin-singlet truncation variety embedded into $\PP\operatorname{Sym}_2(\wedge^2 \CC^4)$ is the Veronese square of the Grassmannian:
    $$
    V_{\{1\}}^{\operatorname{SU}(2)} = \nu_2(\operatorname{Gr}(2,4)) \subseteq \PP\operatorname{Sym}^2(\wedge^2\mathbb C^{4}).
    $$ 
    In particular, $V_{\{1\}}^{{\rm SU}(2)}$ has dimension $4$ and degree $32$. Notice that the ambient singlet space $\mathcal H_{4}^{{\rm SU}(2)}$ has projective dimension $19$, whereas
$\operatorname{Sym}^2(\wedge^2\mathbb C^{4})$ has projective dimension $20$.
The difference is accounted for by the Pl\"ucker relation
\[
p_{12}p_{34} - p_{13}p_{24} + p_{23}p_{14} = 0,
\]
which becomes a \emph{linear} relation among the Veronese coordinates. We can therefore project onto $\PP\mathcal{H}_4^{\operatorname{SU}(2)}$ and embed the truncation variety $\nu_2({\rm Gr}(2,4))$ into that space.
\end{example}

This example generalizes to arbitrary $m$ and $k$:

\begin{theorem}\label{thm:vgr}
    The singlet CCS truncation variety
is isomorphic to the Veronese square of the Grassmannian, specifially
$
V_{\{1\}}^{{\rm SU}(2)} \;\cong\; \nu_2\bigl(\operatorname{Gr}(k,m)\bigr).
$
In particular,
\begin{equation}
\dim V_{\{1\}}^{{\rm SU}(2)} = k(m-k),
\quad {\rm and} \quad 
\deg V_{\{1\}}^{{\rm SU}(2)} = 2^{\,k(m-k)}\,\deg\bigl(\operatorname{Gr}(k,m)\bigr).
\end{equation}
\end{theorem}

\begin{proof}
    Recall that the Plücker embedding parameterizes the  Grassmannian $\operatorname{Gr}(k,m)$ by the maximal minors of a $k \times m$ matrix $m(t) = \begin{pmatrix}
        I_k & t
    \end{pmatrix}$. Here $t = (t_{i,b})_{1 \le i \le k < b \le m}$ is a $k \times (m - k)$ matrix with the CCS cluster amplitudes as entries. We denote the \textit{Plücker coordinates} by $p_I(t) = \det(m(t)_I)$ where $I \subseteq [m]$ and $|I| = k$.
    
    We consider the composition of the spin-singlet exponential map with the natural inclusion $\mathcal H_{d}^{{\rm SU}(2)}\hookrightarrow \mathcal H_{d}\cong  \CC^{\binom{n}{d}}$. Its coordinates are the maximal minors of a $d \times n$ matrix $M(t)$, analogues to Eq.~\eqref{eq:grparam}.
    The columns of this matrix are indexed by the spin orbitals, i.e., $[\![m]\!]$, and the rows by the occupied spin orbitals, i.e., $[\![k]\!]$. 
    Row and column permutation such that spin-up orbitals precede spin-down orbitals yields a block-diagonalization of $M(t)$, i.e.

    $$
    M'(t) = 
    \begin{pmatrix}
        m(t) & 0_{k,m - k}\\
        0_{k,m - k} & m(t)
    \end{pmatrix}.
    $$
    The maximal minors of $M'(t)$ can be computed using $k \times k$ Laplace expansion. First, we note that a $2k\times 2k$ minor of $M'(t)$ is nonzero only if it selects exactly $k$ columns from the spin-down block and $k$ columns from the spin-up block. Now for $I,J\in\binom{[m]}{k}$ we get
\begin{equation}
\det\bigl(M'(t)_{\{I\downarrow\}\cup\{J\uparrow\}}\bigr)
=
\det\bigl(m(t)_I\bigr)\,\det\bigl(m(t)_J\bigr)
=
p_I(t)\,p_J(t).
\end{equation}
    These are the coordinates of the composed map. We have now embedded the spin-singlet truncation variety into the full space of quantum state, and want to project down again to the space of ${\rm SU}(2)$--invariants. To that end, we embed the image into the space $\PP\operatorname{Sym}_2(\wedge^k \CC^m)$ -- as is described in Remark \ref{re:isspin}. We obtain the Veronese embedding of~$\operatorname{Gr}(k,m)$:
    $$
    \CC^{k \times (m - k)} \to \PP^{\binom{m}{k} - 1} \to \PP\operatorname{Sym}_2(\wedge^k \CC^m), \quad t \mapsto (p_I(t) p_J(t))_{I \le J \in \binom{[m]}{k}}.
    $$
    By Proposition \ref{prop:spin_adapted_exp} the spin-singlet exponential parameterization maps into the invariant space $\mathcal{H}_{2k}^{\operatorname{SU}(2)}$, so we can further embed the Veronese square of the Grassmannian into $\PP\mathcal{H}_{2k}^{\operatorname{SU}(2)}$. Note that there are linear relations among the coordinates, namely, the Plücker relations -- quadratic equations in the $p_I(t)$ that cut out the Grassmannian. There are exactly 
    $$
    \dim(\operatorname{Sym}_2(\wedge^k \CC^m)) - \dim(\mathcal{H}_{2k}^{\operatorname{SU}(2)}) = \binom{\binom{m}{k} + 1}{2} - N(m + 1, k + 1).
    $$
    Plücker relations, minimally generating the Grassmannian. The projection of the Veronese square into $\PP\mathcal{H}_{2k}^{\operatorname{SU}(2)}$ is therefore also an embedding of $\nu_2(\operatorname{Gr}(k,m))$.
\end{proof}

We notice the Grassmannian of lines in $m$--dimensional space is a projective space, namely $\operatorname{Gr}(1,m) \cong \PP^{m - 1}$, and therefore linear. The CCS spin-singlet truncation varieties corresponding to one electron pair, $k = 1$, are therefore the \textit{2nd Veronese varieties}, $V_{\{1\}}^{{\rm SU}(2)} \cong \nu_2(\PP^{m - 1})$. 

\subsection{Spin-Adapted Coupled Cluster Equations} 

The truncation varieties are used to approximate the solutions of the Schrödinger equation Eq.~\eqref{eq:schr}. 
We fix a level set $\sigma$ and consider the spin-singlet truncation variety $V_\sigma^{{\rm SU}(2)}$. In general, none of the quantum states on $V_\sigma^{{\rm SU}(2)}$ will be eigenvectors of the Hamiltonian. We then relax the constraints of the eigenvalue problem and define the coupled cluster equations:
\begin{equation}\label{eq:CCeqs}
    (\hat H \psi)_\sigma = E\psi_\sigma, \quad \psi \in V_\sigma^{{\rm SU}(2)}.
\end{equation}
Here $\psi_\sigma$ denotes the projection of quantum state $\psi$ onto coordinates $\psi_I$ with excitation level $|I \backslash [\![k]\!]| \in \sigma$. Hence, $(\cdot)_\sigma$ defines a projection from $\mathcal{H}_{2k}$ onto the graded subspace $\oplus_{\ell \in \sigma} H_{2k}^{(\ell)}$.

For generic parameters $h_{pq}$ and $v_{pqrs}$, the number of solutions to the CC equations in Eq.~\eqref{eq:CCeqs} is finite and fixed. We call this number the \textit{coupled cluster degree} of the truncation variety $V_\sigma^{{\rm SU}(2)}$ and denote it by $\operatorname{CCdeg}(V_\sigma^{{\rm SU}(2)})$. By the parameter continuation theorem \cite[Theorem~1]{morgan1989coefficient}, the CC degree serves as an upper bound to the number of solutions of Eq.~\eqref{eq:CCeqs} for any Hamiltonian operator $\hat H$. 
Computing the CC degree of a truncation variety is generally difficult, as it requires solving the CC equations for a generic Hamiltonian $\hat H$. Nevertheless, one can obtain an upper bound in terms of invariants of $V_\sigma^{{\rm SU}(2)}$: 

\begin{proposition}\label{prop:CCbound}
    The coupled cluster degree of $V_\sigma^{{\rm SU}(2)}$ fulfills the following inequality:
    $$
    \operatorname{CCdeg}(V_\sigma^{{\rm SU}(2)}) \le (\dim(V_\sigma^{{\rm SU}(2)}) + 1)\deg(V_\sigma^{{\rm SU}(2)}).
    $$
\end{proposition}

The proof of this Proposition is omitted as it is analogous to the proof of \cite[Theorem 5.2]{faulstich2024algebraic}. The dimension of a truncation variety can easily be computed by counting basis vectors with excitation level in $\sigma$. For example, the CCS and CCD spin-singlet truncation varieties have the dimensions
$$
\dim(V_{\{1\}}^{{\rm SU}(2)}) = k(m - k), \quad \dim(V_{\{2\}}^{{\rm SU}(2)}) = \binom{k(m - k) + 1}{2},
$$
respectively, and the dimension of the CCSD truncation variety is their sum. Finding an explicit formula for the degree of a truncation variety is more challenging. In the CCS case, an explicit formula was obtained in Theorem~\ref{thm:vgr}. By contrast, for CCD and CCSD no such formula is known, and the degree for specific values of $k$ and $m$ must be determined numerically. In some cases, we can obtain explicit formulas for the CC degree itself, such as:

\begin{theorem}
    Let $k = 1$. The CC degree of the 2nd Veronese variety $V_{\{1\}}^{{\rm SU}(2)}  \cong \nu_2(\PP^{m - 1})$~is
    $$
    \operatorname{CCdeg}(\nu_2(\PP^{m - 1})) = 2^{m} - 1.
    $$
    The CC and ED degrees of the Veronese variety coincide, see \cite{breiding2024metric}.
\end{theorem}
\begin{proof}
    Here $k = 1$ and so the dimension of the invariant space simplifies to 
    $$
    \dim(\mathcal{H}_{2k}^{{\rm SU}(2)}) = N(m + 1, 2) = \binom{m + 1}{2}.
    $$
    Remark~\ref{re:isspin} implies that, in this case, $\mathcal{H}_2^{\operatorname{SU(2)}} \cong \operatorname{Sym}_2(\CC^m)$, so the basis vectors of $\mathcal{H}_2^{{\rm SU}(2)}$ are $e_{i \downarrow, i \uparrow}$ and
    $e_{i\, \downarrow\,, j\, \uparrow} + e_{i\, \uparrow\,, j\, \downarrow}$. They are indexed by subsets in $[m]$ of size one and two. Also, the Veronese square $x^{\otimes 2} \in \nu_2(\PP^{m - 1})$ of $x \in \PP^{m - 1}$ projects onto $x$ under truncation, that is $(x^{\otimes 2})_{\{1\}} = x$.
    We can therefore rewrite the CC equations as
    \begin{equation}\label{eq:specialCCeqs}
         H_{\{1\}} x^{\otimes 2} = E x, \qquad x \in \PP^{m - 1},
    \end{equation}
    where $ H_{\{1\}}$ is a $m \times \binom{m + 1}{2}$ submatrix of the generic Hamiltonian $\hat H$, with only the rows indexed by singletons in $[m]$. We can construct a generic $m \times m \times m$ symmetric tensor $\eta$ such that $ H_{\{1\}}$ has entries $ H_{i, j\ell} = \eta_{i,j} \eta_{i,\ell}$. The tensor $\eta$ uniquely corresponds to a homogeneous polynomial of degree three, given by $F_\eta = \langle \eta, x^{\otimes 3} \rangle$, and its gradient satisfies $\nabla F_\eta =  H_{\{1\}} x^{\otimes 2}$. 
    The CC equations in Eq.~\eqref{eq:specialCCeqs} are therefore precisely the fixed-point equations of $\nabla F_h$. Equivalently, they are the eigenpair equations for the symmetric tensor $\eta$. Since $\hat H$ and $\eta$ are generic tensors we get by \cite[Theorem 12.17]{breiding2024metric} that the CC degree is $2^m - 1$.
\end{proof}

In Eq.~\eqref{eq:CCeqs} we formulate the unlinked coupled cluster equations in an algebraic way. This might be unintuitive for readers that are experienced in electronic structure theory. Nonetheless this formulation is important to prove results such as Proposition \ref{prop:CCbound}. 
However, this presentation is not optimal for explicitly solving the CC equations.
Since the restriction of the  exponential map to $\mathcal{V}_\sigma$ is injective, we can write the unlinked coupled cluster equations in the following way:
\begin{equation}\label{eq:compCCeqs}
    ((\hat H - E \cdot \hat I_{\mathcal{F}}) \exp(\hat T_\sigma) e_{[\![k]\!]})_\sigma  = 0.
\end{equation}
This is a square polynomial system of size $\dim(V_\sigma^{{\rm SU}(2)}) + 1$, whose variables are the energy $E$ and the cluster amplitudes $t_{I,B}$ with excitation level $|I| \in \sigma$. Using this formulation, we can solve the CC equations with \texttt{HomotopyContinuation.jl}~\cite{breiding2018homotopycontinuation}. For an in-depth description on how to solve the coupled cluster equations, we refer the reader to~\cite{sverrisdottir2024exploring}

\medskip
We close this chapter with a short remark on L\"owdin's symmetry dilemma:

\begin{remark}
\label{re:lowdin_dilemma}
Although the Hamiltonian is ${\rm SU}(2)$-invariant and its exact eigenstates can be chosen as simultaneous eigenstates
of $\hat S^2$ and $\hat S_z$, approximate ans\"atze may face a trade-off commonly referred to as \emph{L\"owdin's symmetry dilemma}.
At the mean-field level, enforcing spin symmetry (e.g.\ restricted Hartree--Fock) yields a spin-pure reference but may give
qualitatively poor energies in regimes of strong (static) correlation (for instance, along bond dissociation), because a single
symmetry-adapted determinant cannot represent near-degeneracy effects. Allowing symmetry breaking (e.g.\ unrestricted Hartree--Fock)
often lowers the energy and captures part of the static correlation, but produces \emph{spin-contaminated} states that do not lie in
$\mathcal H_d^{{\rm SU}(2)}$. 
The same tension can propagate to correlated methods built on a mean-field reference. In \emph{restricted} coupled cluster theory we
deliberately formulate the equations on the ${\rm SU}(2)$--invariant sector, see Eq.~\eqref{eq:schr}, so that the CC state is
spin-adapted by construction; however, this restriction can necessitate higher excitation rank to recover correlation effects that
a broken-symmetry reference may mimic at lower rank. Alternative strategies include symmetry restoration (spin projection) applied to
broken-symmetry references, which aims to combine the energetic flexibility of symmetry breaking with a final spin-pure wavefunction.
\end{remark}

\section{Numerical Simulations}\label{sec:numsim}

A practical algebraic measure of computational complexity in CC theory is the \emph{number of isolated solutions} (counted with multiplicity) for a generic instance of the CC equations, i.e., the \emph{CC degree}. In homotopy-based solvers, the CC degree directly controls the number of solution paths that {\it must} be tracked, and is therefore a reliable proxy for runtime and memory. The central numerical message of this section is that imposing ${\rm SU}(2)$--spin adaptation reduces the CC degree by orders of magnitude, making exhaustive solution and continuation computations feasible in regimes where the spin-generalized formulation becomes intractable.
The simulations presented here were performed using the software packages \texttt{HomotopyContinuation.jl}~\cite{breiding2018homotopycontinuation} and \texttt{PySCF}~\cite{sun2015libcint,sun2018pyscf,sun2020recent}.

\subsection{Scaling of Generic Spin Restricted Systems}
We begin by comparing the number of roots of the spin restricted coupled cluster equations with the number of roots of the spin generalized coupled cluster equations from~\cite[Section~5]{faulstich2024algebraic}. For $k=1$ (and $d = 2$) we investigate how the CC degree scales with the number of spin orbitals $n$, both for the CC equations for singles ($\sigma=\{1\}$) and for doubles ($\sigma=\{2\}$), similar to~\cite[Example~6.3]{faulstich2024algebraic}. In Figure~\ref{fig:scaling}, the blue curve shows the spin-generalized CC degree (see Theorem~5.5 and Corollary~5.3 in~\cite{faulstich2024algebraic} for CCS and CCD, respectively),
while the orange curve shows the corresponding spin-restricted CC degree. 

\begin{figure}[h!]
 \centering
 \begin{subfigure}[b]{0.45\textwidth}
     \centering
     \includegraphics[width=\textwidth]{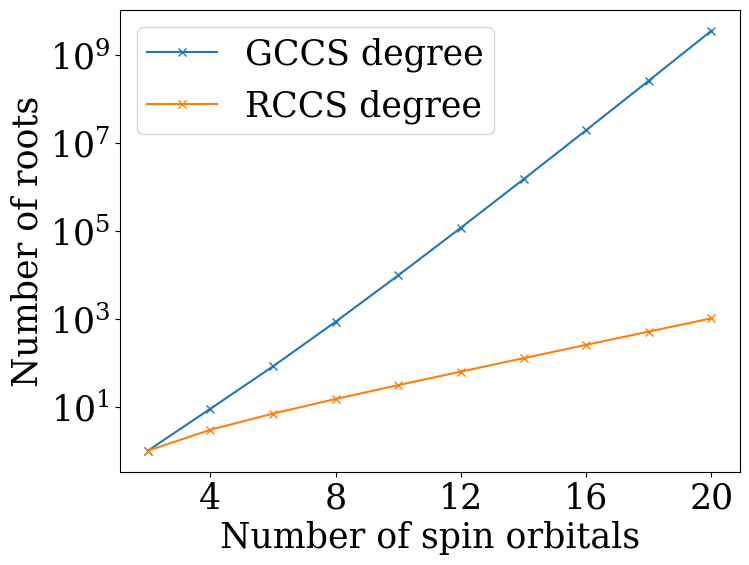}
     \label{fig:scaling_ccs}
 \end{subfigure}
 \hfill
 \begin{subfigure}[b]{0.45\textwidth}
     \centering
     \includegraphics[width=\textwidth]{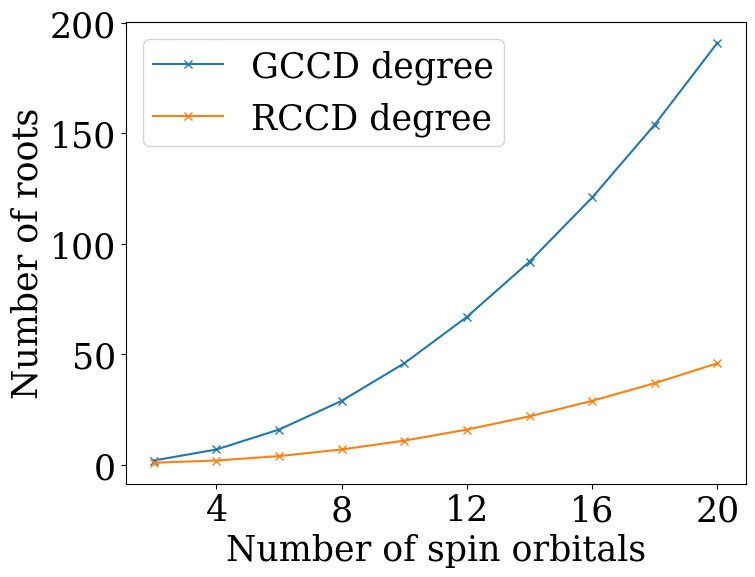}
     \label{fig:scaling_ccd}
 \end{subfigure}\vspace{-0.2in}
 \caption{\label{fig:scaling}
 Comparison of the number of roots between spin restricted and spin generalized CC equations for CCS (left) and CCD (right) at $k=1$.
 The reduction in degree translates into a corresponding reduction in the number of solution paths that must be tracked.}
\end{figure}

Figure~\ref{fig:scaling} quantifies the reduction in generic solution complexity obtained by imposing ${\rm SU}(2)$--symmetry.
For $k=1$ and CCS, the spin-generalized degree grows to $\sim 10^{9}$ at $20$ spin orbitals, whereas the spin-restricted degree is
only $\sim 10^{3}$, corresponding to a reduction of roughly six orders of magnitude in the number of isolated solutions (and hence in
the number of continuation paths in homotopy-based solvers). For CCD the degrees are smaller but the same trend persists: at $20$ spin
orbitals the degree decreases from about $190$ (generalized) to about $45$ (restricted), a reduction by a factor of $\approx 4$.
These data indicate that spin adaptation yields a substantial algorithmic simplification of the CC polynomial systems
already in the minimal-electron setting, motivating its use in the numerical studies below.

\medskip
Indeed, we observe that the effect becomes numerically pronounced and computationally observable already for systems as small as four electrons. We compare CCD degrees for the spin-restricted and
spin-generalized formalisms when $d=4$ (two electron pairs, $k = 2$) and the number of spin-orbitals is $n = 8,9,10,11,12$. The spin-generalized CCD degrees are taken from
\cite[Section~5]{faulstich2024algebraic} where available, and for larger $n$ we report the best available upper bounds from~\cite{faulstich2026On}. The resulting degrees are summarized in Table~\ref{tab:RCCDvsGCCDd4}.

\begin{table}[ht!]
    \centering
    \begin{tabular}{c|ccccc}
        $n$ & 8 & 9 & 10 & 11 & 12 \\
        \hline
        RCCD & 20 & -- & 998 & -- & $\sim 1.19\cdot 10^7$\\
        GCCD & 72 & 1,823 & 2,523,135 & $\le 3.2\cdot 10^{11}$  & $\le 1.18\cdot 10^{23}$
    \end{tabular}
    \caption{RCCD and GCCD degrees for four electrons ($d=4$). For $n=11,12$ the GCCD entries are upper bounds.
    The collapse in degree under spin adaptation directly reflects a collapse in the number of solution branches.}
    \label{tab:RCCDvsGCCDd4}
\end{table}

Two features are noteworthy. First, even at modest size ($n=10$), the degree collapses from $2{,}523{,}135$ in GCCD to $998$ in RCCD,
a reduction by more than three orders of magnitude. Second, the available bounds indicate that the gap continues to widen rapidly as
$n$ grows, underscoring that spin adaptation is not a cosmetic symmetry constraint but a decisive computational simplification. We now illustrate how the incorporation of spin-adaptation pushes the boundaries of algebro computational exploration for chemical systems.

\subsection{Lithium Hydride}

Lithium hydride (LiH) has served as a small yet ``chemically realistic'' benchmark system for algebraic--computational studies.
In a minimal-basis (STO-6G) description, LiH is modeled with four electrons (two electron pairs) in six spatial orbitals, i.e., twelve spin
orbitals. The minimal spatial basis is given by the atomic orbitals
\[
\mathrm{H}\,1s,\quad \mathrm{Li}\,1s,\quad \mathrm{Li}\,2s,\quad \mathrm{Li}\,2p_x,\quad \mathrm{Li}\,2p_y,\quad \mathrm{Li}\,2p_z.
\]
Moreover, LiH exhibits a clear energetic separation between core and valence orbitals, as reflected in the
molecular-orbital energies obtained from spin-restricted Hartree--Fock calculations~\cite{szabo2012modern}
\[
(-2.51132217,\,-0.3024674,\,0.06426691,\,0.14495239,\,0.14495239,\,0.59035732).
\]
The large gap to the lowest orbital is consistent with a chemically inert core, suggesting that LiH is well described by an
active-space picture in which correlation is dominated by the valence pair, while core excitations are energetically suppressed.
In the frozen core approximation we expect CCSD to be exact since this corresponds to a two-electron problem. 

\subsubsection{LiH Dissociation in a $\sigma$-Active Space}

We may choose the Li$-$H bond axis as the $z$-axis of the coordinate frame. In this setting the electronic structure separates into $\sigma$ and $\pi$ symmetry subspaces: the $\mathrm{Li}\,2p_x$ and $\mathrm{Li}\,2p_y$ atomic orbitals transform as $\pi$ functions (perpendicular to the bond) and therefore do not mix with the $\sigma$ manifold that governs bonding. Consequently, the molecular-orbital coefficient matrix is block diagonal, with a decoupled $\{2p_x,2p_y\}$ block. We therefore restrict our simulations to the {\it $\sigma$-active space} spanned by $\mathrm{H}\,1s$, $\mathrm{Li}\,1s$, $\mathrm{Li}\,2s$, and $\mathrm{Li}\,2p_z$, which yields an effective model with two electron pairs and four spatial orbitals. Note that if the linear symmetry were broken along the reaction path, this decoupling would no longer be exact and small $\sigma$--$\pi$ mixing could occur. Numerically, this decoupling is directly evident in the molecular-orbital (MO) coefficient matrix $C$ from spin-restricted Hartree--Fock calculations~\cite{szabo2012modern} (exemplified at the LiH-equilibrium geometry): 
\[
C=
\left[
\begin{array}{ccc|cc|c}
 0.974 & -0.328 & -0.147 & 0 & 0 & -0.385\\
-0.009 &  0.370 &  0.881 & 0 & 0 & -0.952\\
-0.020 &  0.441 & -0.548 & 0 & 0 & -1.076\\ \hline
 0     &  0     &  0     & 0 & 1 & 0\\
 0     &  0     &  0     & 1 & 0 & 0\\ \hline
 0.072 &  0.535 & -0.183 & 0 & 0 & 1.557
\end{array}
\right],
\]
where two of the orbitals are supported entirely on the $\{ \mathrm{Li}\,2p_x, \mathrm{Li}\,2p_y\}$ basis functions and have zero coefficients on the remaining atomic orbitals. Hence, $C$ can indeed be arranged into a block diagonal matrix with a decoupled $2\times 2$ $\pi$-block. We investigate the root structure for RCCSD while dissociating LiH over bond lengths from $1$ to $3$ bohr. The spin-singlet CC degree in this case is
\[
\operatorname{CCdeg}(V_{\{1,2\}}^{\rm SU(2)})=620.
\]
opposed to the spin-generalized CC degree reported in~\cite{sverrisdottir2024exploring}:
\[
\operatorname{CCdeg}(V_{\{1,2\}}) \approx 16,952,996.
\]
We discretize the interval using an equidistant grid with 50 points and compute between $618$ and $620$ RCCSD solutions per grid point. For comparison, diagonalizing the corresponding Hamiltonian (i.e., a $20\times 20$ real-valued symmetric matrix) yields $20$ distinct eigenvalues.

\begin{figure}[htpb]
     \centering
     \begin{subfigure}[b]{0.45\textwidth}
         \centering
         \includegraphics[width=\textwidth]{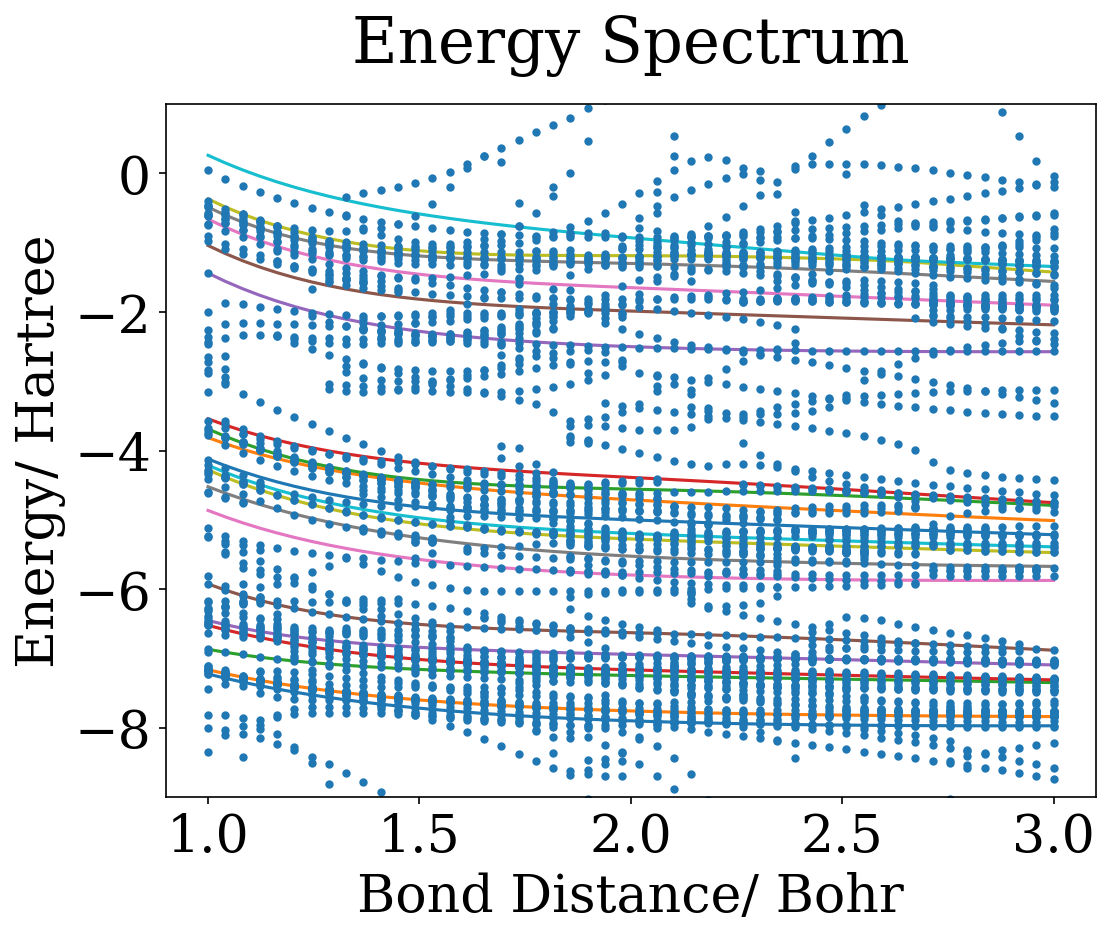}
         \caption{\label{fig:LiHCCSDfull}}
     \end{subfigure}
     \hfill
     \begin{subfigure}[b]{0.45\textwidth}
         \centering
         \includegraphics[width=\textwidth]{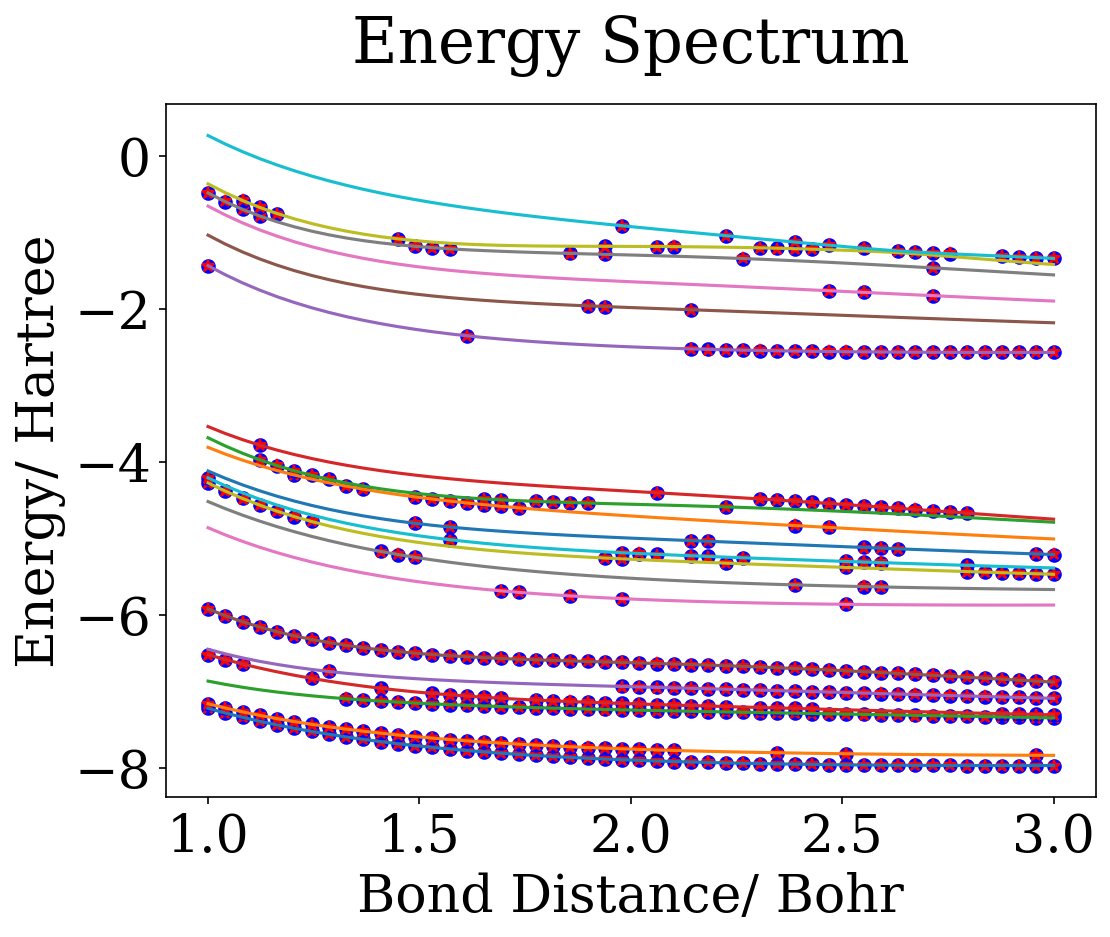}
         \caption{\label{fig:LiHCCSDmin}}
     \end{subfigure}
     \caption{\label{fig:LiHCCSD}
     LiH dissociation in a minimal basis ($k=2$, $m=4$): comparison of RCCSD solution branches with the exact eigenvalue curves.
     (a) All RCCSD solution branches compared to the exact spectrum. (b) RCCSD solutions lying near an eigenvalue (the physically relevant).
     }
\end{figure}

Figure~\ref{fig:LiHCCSDfull} shows all RCCSD solution branches together with the exact eigenvalue curves, while
Figure~\ref{fig:LiHCCSDmin} highlights only those RCCSD solutions that lie close to an eigenvalue,i.e., an energy difference of less than $10^{-3}$.

\subsection{All-electron LiH calculations}

In previous algebro-computational investigations, all-orbital simulations for LiH were out of computational reach due to the sheer size of the CC polynomial system and the corresponding solution count. As reported in~\cite{faulstich2024algebraic}, the corresponding CCD and CCSD systems consist of $168$ and $200$ equations, respectively -- well beyond the practical limits of current algebraic numerical solvers. Moreover, the number of homotopy paths that would need to be tracked vastly exceeds the approximately $1.7 \times 10^7$ paths tracked in~\cite{sverrisdottir2024exploring}, where the CCSD calculations for LiH in $8$ spatial orbitals took over $30$ days.
In the present setting, however, imposing spin adaptation reduces the effective algebraic complexity sufficiently to make a RCCD solution-landscape computation feasible, for the all-electron system.  

\medskip
All subsequent computations were done on the MPI--MiS computer server, using four 18-Core Intel Xeon E7-8867 v4 at 2.4 GHz (3072 GB RAM).
For the generic (random-coefficient) instance of the all-electron RCCD system, the monodromy solver required 21 days
to recover the complete solution set, revealing the RCCD degree for $k = 2$ and $m = 6$ to be
\begin{equation}\label{eq:CCdegLiH}
    {\rm CCdeg}(V_{\{1,2\}}^{\rm SU(2)}) \approx 11{,}920{,}113
\end{equation}
Recall that a lower bound for the same system in GCCD is $10^{23}$, see Table~\ref{tab:RCCDvsGCCDd4}.
Once the generic solutions were obtained, a parameter homotopy tracked them to the LiH instances in 3 days and 48 minutes. This computation produced $67{,}909$ solutions, of which 552 are non-singular and 106 are real.
In Figure \ref{fig:LiHCCDfull} we compare the full solution spectrum with the 71 distinct eigenvalues of the symmetric $105 \times 105$ Hamiltonian matrix. In Figure \ref{fig:LiHCCDkept} we highlight the only those RCCD solutions that lie within a $10^{-3}$ radius of an eigenvalue.
In total, this provides (to our knowledge) the first all-electron LiH computation in which the \emph{entire} RCCD solution set is explicitly resolved.

\begin{figure}[h!]
     \centering
     \begin{subfigure}[b]{0.45\textwidth}
         \centering
         \includegraphics[width=\textwidth]{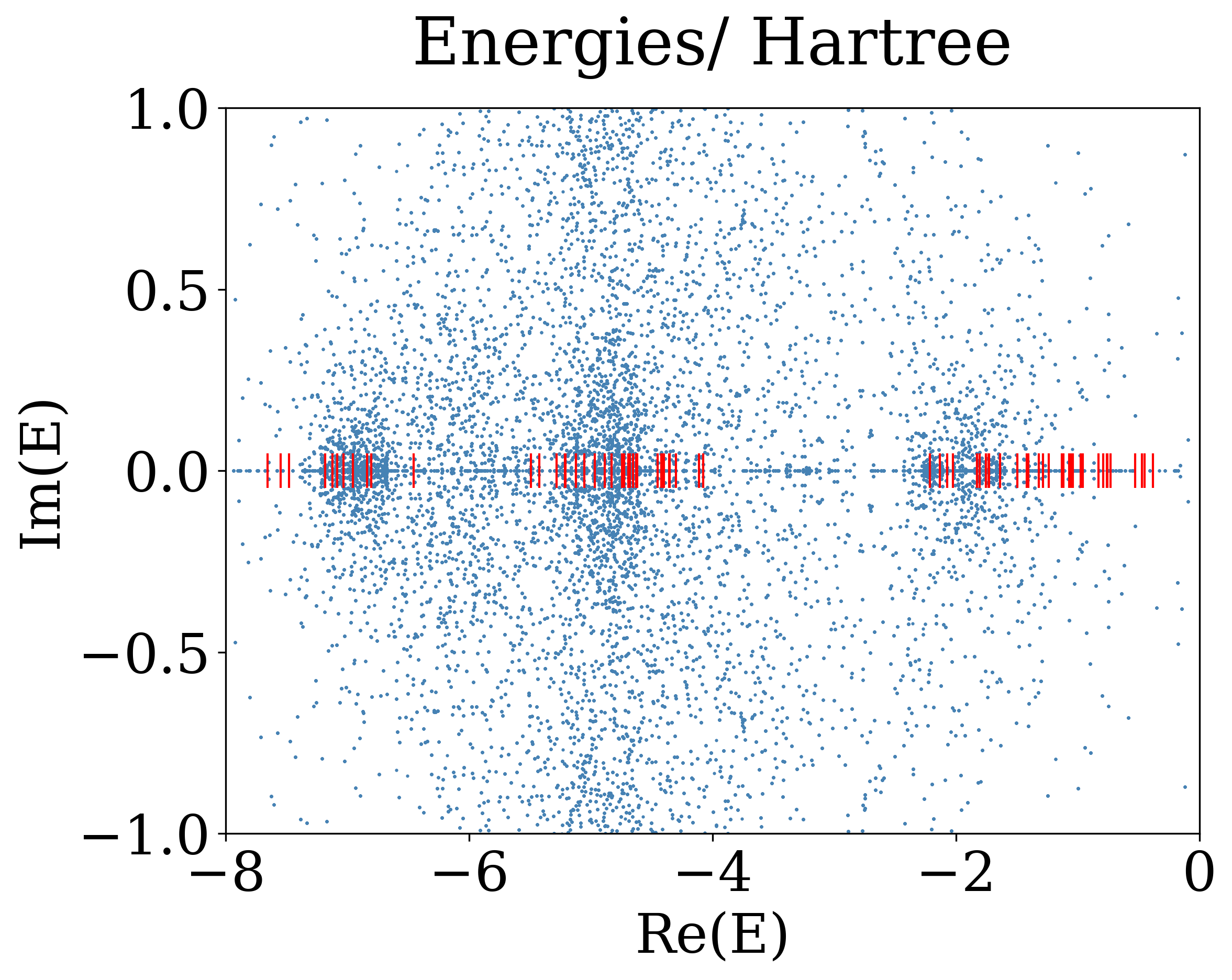}
         \caption{\label{fig:LiHCCDfull}}
     \end{subfigure}
     \hfill
     \begin{subfigure}[b]{0.45\textwidth}
         \centering
         \includegraphics[width=\textwidth]{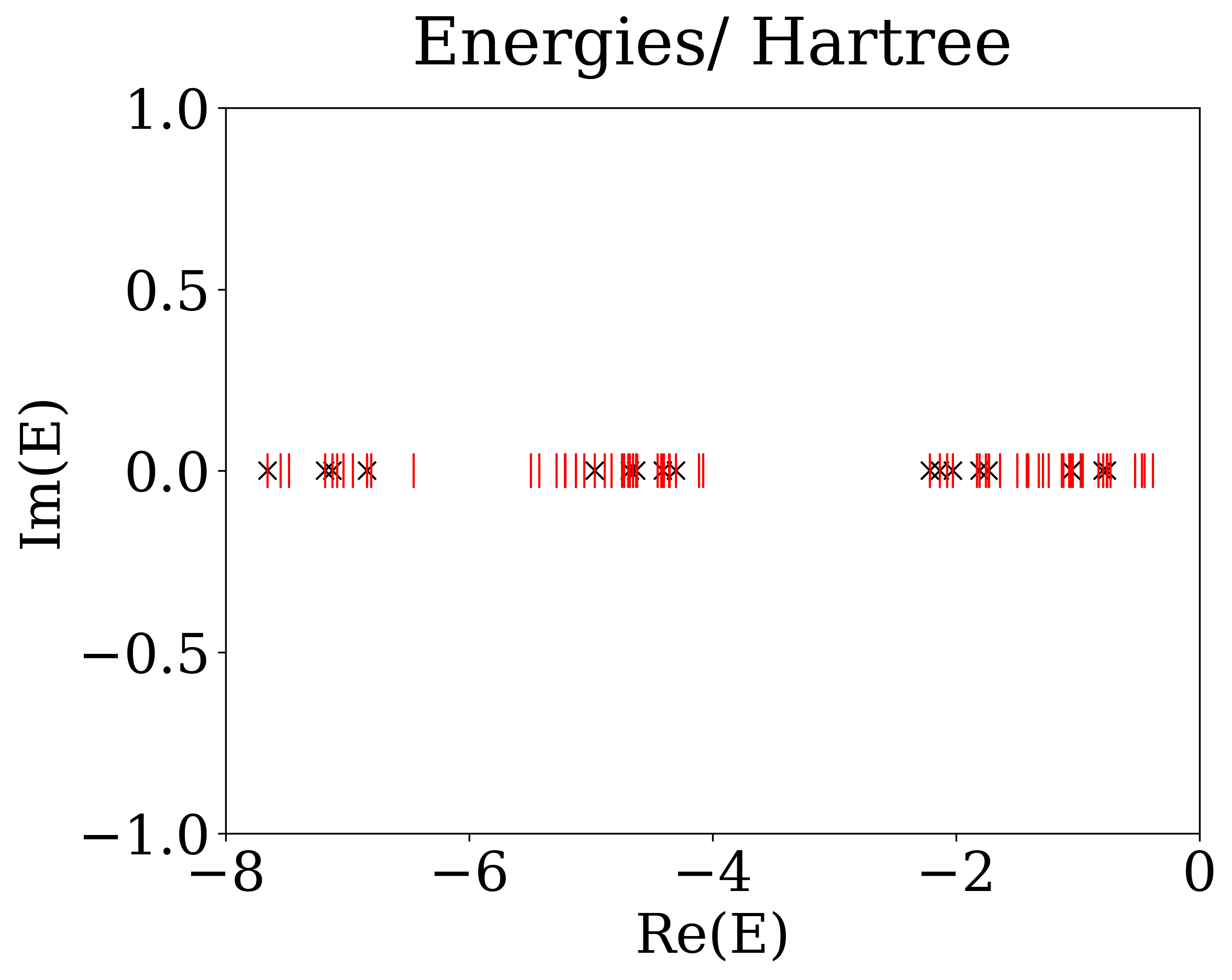}
        \caption{\label{fig:LiHCCDkept}}
     \end{subfigure}
     \caption{\label{fig:LiHCCSD}
     LiH dissociation in a full basis ($k=2$, $m=6$): comparison of RCCD solutions with the exact eigenvalues.
     (a) All RCCSD solutions compared to the exact spectrum. (b) RCCD solutions lying near an eigenvalue (physically relevant).
     }
\end{figure}

\medskip
The resulting all-electron picture reinforces two qualitative observations already present in the $\sigma$-active-space study. First, there are multiple physically relevant RCCD solutions close to exact eigenvalues. Second, the RCCD polynomial system admits additional, non-physical branches whose energies do not correspond to any eigenvalue of the Hamiltonian and may even fall below the exact ground-state energy. Such ``overcorrelated'' solutions are consistent with the non-variational character of coupled cluster theory, but here they appear as a \emph{structural} feature of the RCCD algebraic solution set rather than a numerical accident. Their systematic presence in the all-electron calculation emphasizes the value of global solution methods (monodromy/parameter homotopy) for characterizing the full CC landscape and motivates {\it a posteriori} diagnostics for identifying the physically meaningful root among many mathematically valid ones~\cite{janssen1998new,nielsen1999double,bartlett2020index,faulstich2023s}.

\subsection{Water}

Water (H$_2$O) in a minimal basis set discretization has ten electrons (five electron pairs) in seven spatial orbitals, i.e., 14 spin orbitals. The minimal spatial basis is given by the atomic orbitals  
\[
\mathrm{H}\,1s,\quad \mathrm{H}\,1s,\quad \mathrm{O}\,1s,\quad \mathrm{O}\,2s,\quad \mathrm{O}\,2p_x,\quad \mathrm{O}\,2p_y,\quad \mathrm{O}\,2p_z.
\]
We note that water exhibits a clear energetic separation between the core and valence manifolds at the Hartree--Fock level. The MO energies (in $E_h$) are
\[
(-20.5043,\,  -1.2759,\,  -0.6209,\,  -0.4593,\, -0.3975,\,   0.5970,\,   0.7299 ).
\]
showing a large gap of $\Delta\varepsilon \approx 19.23\,E_h$ ($\approx 523\,\mathrm{eV}$) between the lowest orbital and the remaining valence/virtual spectrum. This separation is consistent with a chemically inert core orbital, so that electron correlation relevant to bonding and low-energy excitations is dominated by the valence space, while core excitations are energetically suppressed.

\medskip
To quantify the core character, we compute a Mulliken-like gross population of molecular orbital $i$ on the chosen core atomic-orbital (AO)  subspace~\cite{mulliken1955electronic,mulliken1955electronic3},
\[
w_{\mathrm{core}}^{(i)} \;=\; \sum_{\mu \in \mathrm{core}} C_{\mu i}\,(SC)_{\mu i},
\]
where $C$ is the MO coefficient matrix and $S$ is the AO overlap matrix. Taking the core subspace to be the oxygen $1s$ AO(s), we obtain for the lowest MO
\[
w_{\mathrm{O\,1s}}^{(0)} = 0.996559,
\]
indicating that MO~0 is essentially a pure O $1s$ core orbital. We therefore adopt a frozen-core treatment by keeping this doubly occupied orbital inactive and correlating only the remaining (valence) orbitals. In practice, this corresponds to freezing MO~0 (the O~$1s$ orbital) in all post-HF and active-space calculations reported here.

\medskip
We solved a generic instance of the RCCD equations for $k = 4$ and $m = 6$ using a monodromy solver. After running the solver for 6 days and 15 hours, we obtained $11{,}920{,}154$ solutions. By particle–hole symmetry, the corresponding RCCD degree coincides with that for $k = 2$ and $m = 6$. Comparing this count with the approximate CC degree reported in Eq.~\eqref{eq:CCdegLiH}, we find a discrepancy of $41$ solutions. We attribute this discrepancy to numerical error and assume that both monodromy computations failed to detect some solutions. In~\cite[Monodromy]{sverrisdottir2024exploring}, we describe the stopping criteria used in monodromy computations and explain why, by their very nature, such methods make it difficult to verify solution completeness. Consequently, the computed solution counts provide numerical approximations to the CC degree that are always lower bounds, and which with high probability are exact or very close to the true value.
Once the generic solutions were obtained, a parameter homotopy tracked them to the H$_2$O
instances in 13 hours, finding $10{,}628{,}368$ solutions, of which $603$ are singular and $7{,}396$ are real. In Figure \ref{fig:H2OCCDfull} we compare the full solution spectrum with the 105 distinct eigenvalues the nonsingular Hamiltonian matrix. In Figure \ref{fig:H2OCCDkept} we highlight the only those RCCD solutions that lie within a $10^{-3}$ radius of an~eigenvalue.

\begin{figure}[h!]
     \centering
     \begin{subfigure}[b]{0.45\textwidth}
         \centering
         \includegraphics[width=\textwidth]{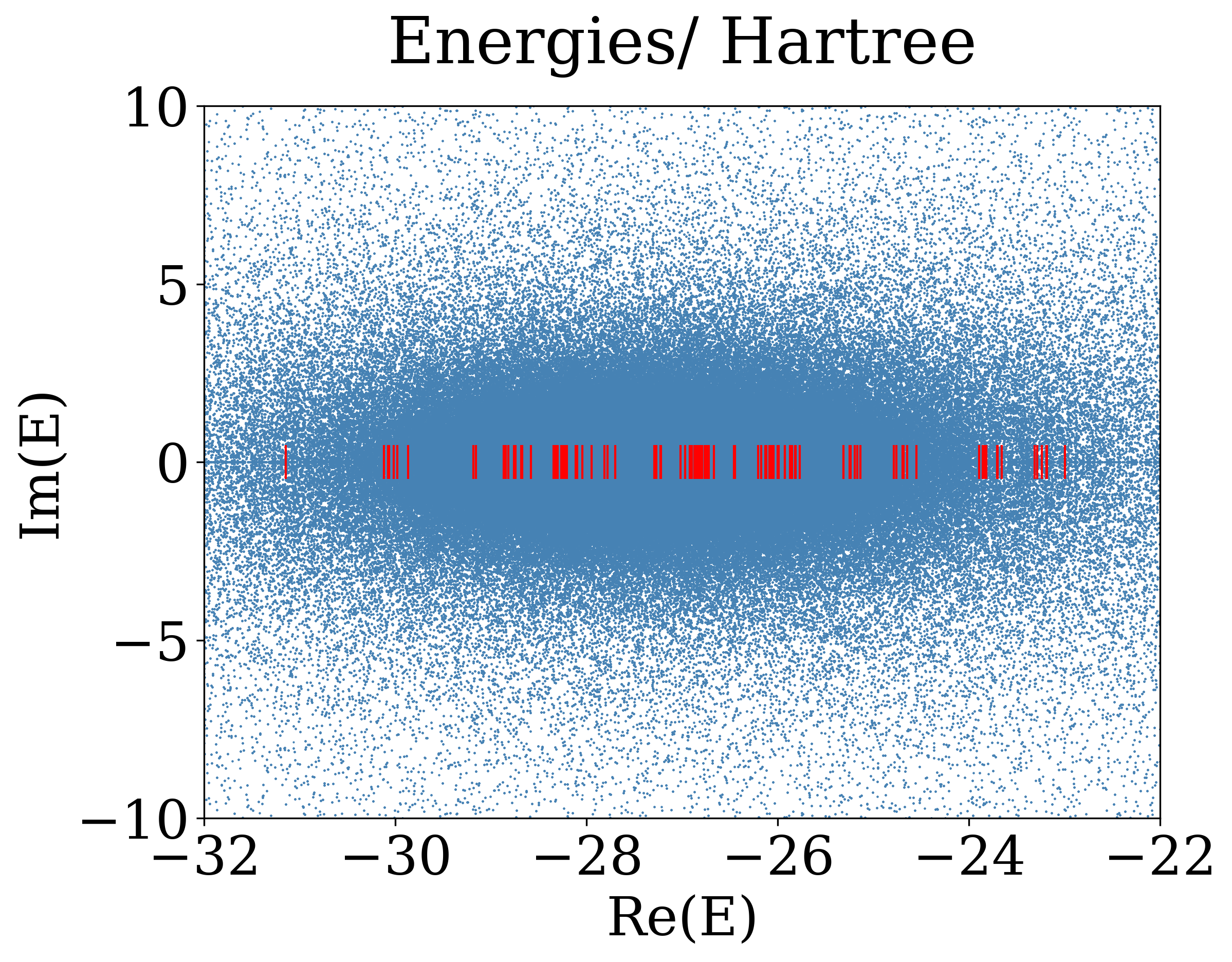}
         \caption{\label{fig:H2OCCDfull}}
     \end{subfigure}
     \hfill
     \begin{subfigure}[b]{0.45\textwidth}
         \centering
         \includegraphics[width=\textwidth]{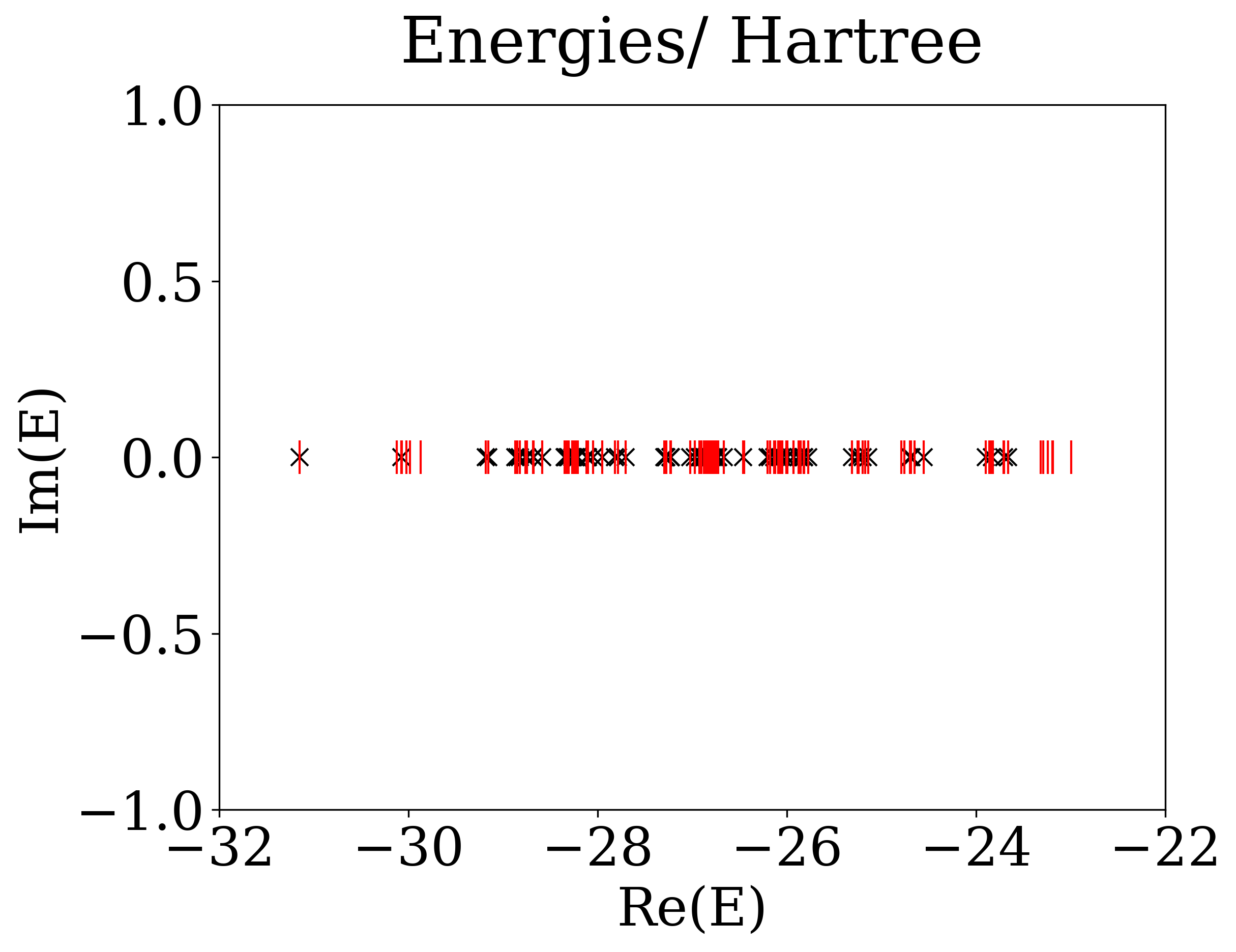}
        \caption{\label{fig:H2OCCDkept}}
     \end{subfigure}
     \caption{\label{fig:LiHCCSD}
     H$_2$O dissociation in a minimal basis ($k=4$, $m=6$): comparison of RCCD solutions with the exact eigenvalues.
     (a) All RCCD solutions compared to the exact spectrum. (b) RCCD solutions lying near an eigenvalue (physically relevant).
     }
\end{figure}

The full RCCD solution spectrum for H$_2$O appears as a cloud surrounding the exact eigenvalues, becoming denser toward the center of the spectrum. Consequently, some solution energies may lie close to an eigenvalue for purely statistical reasons, rather than reflecting physical significance. This indicates that further investigation into the numerical stability of the computed roots is needed. However, this lies beyond the scope of this manuscript and is left for future work.

\medskip

\noindent {\bf Acknowledgment}:
We thank Veronica Calvo Cortes for her helpful insights on representation theory. 


\pagebreak
\bibliographystyle{unsrt}
\bibliography{bibtex.bib}

\end{document}